\documentclass[aps,twocolumn,floatfix,nofootinbib]{revtex4-2}

\usepackage{amsmath,amsthm,amsfonts,amssymb,times,bbm,graphicx,bbold,mathrsfs}
\usepackage[hidelinks, colorlinks=true, allcolors=blue]{hyperref}
\usepackage{float}

\usepackage{multirow}

\newtheorem{theorem}{Theorem}

\newtheorem*{proposition*}{Proposition}
\newtheorem{lemma}[theorem]{Lemma}
\newtheorem*{lemma*}{Lemma}

\newtheorem{definition}[theorem]{Definition}

\newtheorem{result}[theorem]{Result}

\newtheorem{condition}{Condition}

\newcommand{\RR}{\mathbb{R}}
\newcommand{\ZZ}{\mathbb{Z}}
\newcommand{\NN}{\mathbb{N}}
\newcommand{\CC}{\mathbb{C}}

\renewcommand{\S}{\mathcal{S}}
\newcommand{\E}{\mathcal{E}}
\renewcommand{\P}{\mathcal{P}}
\newcommand{\D}{\mathcal{D}}

\renewcommand{\L}{\mathcal{L}}
\renewcommand{\O}{\mathcal{O}}
\newcommand{\B}{\mathcal{B}}
\newcommand{\M}{\mathcal{M}}

\newcommand{\Em}{\mathscr M}
\newcommand{\Kay}{\mathscr K}
\newcommand{\Cee}{\mathscr C}
\newcommand{\Gee}{\mathscr G}

\newcommand{\Ess}{\mathscr S}

\newcommand{\1}{\underline{1}}
\newcommand{\Id}{\mathbb{1}}

\DeclareMathOperator{\tr}{tr}
\DeclareMathOperator{\Span}{span}

\DeclareMathOperator{\img}{img}

\DeclareMathOperator{\conv}{conv}
\DeclareMathOperator{\cone}{cone}

\usepackage{xcolor}

\usepackage{algorithm}
\usepackage{algorithmicx}
\usepackage{algpseudocode}
\algrenewcommand\algorithmiccomment[1]{\hfill {#1}}

\usepackage{stmaryrd}

\begin{document}


\title{Probabilistic theories stable under teleportation}

\author{Lionel J.\ Dmello}
\email{ldmello@thp.uni-koeln.de} 
\author{David Gross} 
\email{david.gross@thp.uni-koeln.de}
\affiliation{ Institute for Theoretical Physics, University of Cologne, Germany } 
\date{Mar 22, 2026}

\begin{abstract}
	A long-standing problem in the foundations of quantum mechanics is to identify a physical principle that explains why algebraically maximal violations of Bell inequalities can generally not be achieved in Nature.
	One recently proposed approach considers \emph{iterated Bell tests},
	where a Bell test is performed on a state that has undergone several rounds of entanglement swapping.
	Obtaining large violations in this scenario is more demanding, because it requires a theory to have both highly entangled states and highly entangled measurements.
	It has been conjectured that the maximal quantum mechanical Clauser-Horne-Shimony-Holt (CHSH)-value of $2\sqrt 2$ might be optimal for any probabilistic theory which, like quantum mechanics, maintains its CHSH-value after an arbitrary number of rounds of entanglement swapping.
	However, in a previous paper, we have exhibited a first example of a probabilistic theory that can sustain a CHSH value of $4$ in this setting.
	In this work, further investigating this property, we give a classification of all general probabilistic theories (GPTs) whose CHSH value is stable in the above sense.
	The problem reduces to a representation-theoretic condition that allows for exactly seven solutions.
	The GPT from our previous work showed some counter-intuitive features, e.g.\ that the local state space had a higher dimension than seemed necessary to realize CHSH tests.
	The classification shows that this is necessarily so.
	Along the way, we generalize the concept of \emph{self-testing} to GPTs.
\end{abstract}

\maketitle 


\section{Introduction}

The search for physical principles that single out quantum mechanics (QM) among descriptions of Nature is a long-standing problem in the foundations of QM.
The seminal work of John S.\ Bell \cite{bell1964einstein}, and the subsequent Clauser-Horne-Shimony-Holt (CHSH) experiment \cite{clauser1969proposed}, established \emph{contextuality} as a principle that rules out classical descriptions of Nature in favour of QM.
So far, realizations of the Clauser-Horne-Shimony-Holt (CHSH) experiment \cite{aspect1982experimental, zeilinger1998violation, rowe2001experimental} have shown that Nature is at least as contextual as predicted by QM, i.e., they achieve a CHSH value of $2\sqrt2$, which is also the maximal value for any quantum strategy (Tsirelson's bound \cite{cirel1980quantum}).
What is still lacking is a physical principle that rules out theories (like boxworld \cite{popescu1994quantum}) which predict a larger CHSH value than QM while still being consistent with the assumptions of the CHSH experiment.
There is a rich body of research exploring this problem \cite{winter2009informationcausality, miklin2021infocausality, fritz2013localorth, hardy2001quantumtheoryreasonableaxioms, chiribella2017quantum, wilce2018royalroad, brassard2006limit, linden2007quantum, navascues2010glance}.
Many such works are based on the following observation: Violations of Bell inequalities predicted by QM tend to be, in general, strictly smaller than the algebraically allowed maximum.
Thus it is conceivable that understanding the origin of this discrepancy leads to a physical principle that explains it.

Weilenmann and Colbeck have recently proposed the \emph{adaptive CHSH game} \cite{weilenmann2020self, weilenmann2020towards} as a task that may explain this discrepancy.
They conjecture that the biggest CHSH value allowed in the adaptive CHSH game is $2\sqrt2$, and that the same value is recovered in the usual CHSH experiment because it is a special case of the former.
The adaptive CHSH game involves performing entanglement swapping before the CHSH test.
Thus it probes the strength of correlations supported by both bipartite states and bipartite effects.
The (convex) sets describing states and measurements are dual to each other.
Thus, expanding one set, in order to reach stronger correlations, inevitably diminishes the other.
In this regard, QM strikes a particular balance, which is why it appears reasonable to conjecture that QM would be optimal for the adaptive CHSH game.

In fact, QM has a stronger property: There exist quantum strategies for which the CHSH value is preserved not only after one round of entanglement swapping, but also after an arbitrary number of rounds.
This property is highly constraining, in the sense that many known examples of general probabilistic theories (GPTs) fail to exhibit it.
However, in our previous work \cite{dmello2024entanglement} we constructed a GPT called oblate stabilizer theory (OST), which not only has this feature, but can also achieve a CHSH value of $4$.
Thus, in view of OST, it is unclear exactly what kind of constraints this property imposes on GPTs.
In this work we deduce these constraints and as a result classify GPTs with this property.

This classification also sheds light on various other aspects such as the resonablilty of the assumption of local tomography and the existence of redundant degrees of freedom in OST.

In the process of obtaining the main result we also generalize the notion of \emph{self-testing} to GPTs.


\subsection{General Probabilistic Theories}\label{subsec:GPTs}

\emph{General probabilistic theories} (GPTs) \cite{segal1947postulates,ludwig1967attempt2,ludwig1968attempt3,dahn1968attempt,stolz1969attempt,stolz1971attempt,davies1970operational,barrett2007information, plavala2023general} is a framework that allows us to describe correlations that are compatible with operationally motivated assumptions (e.g.\ no-signaling), but are not necessarily realized in QM.
The goal of such a framework is to study the observations of QM in a broader setting in order to extract the underlying physical principles governing them.

In this work we adopt the GPT formalism (with minor modifications) from our previous work (Ref.~\cite{dmello2024entanglement}). 
Here we provide an overview of the formalism and the modifications. 
For a more detailed description, please refer to Sec.~I and Sec.~II of \cite{dmello2024entanglement}.
For a general overview of the GPT formalism we refer the reader to \cite{plavala2023general}.
We assume that a GPT is specified by 
\begin{itemize}
    \item A finite-dimensional real vector space $V$.
    \item A collection of closed, pointed convex cones\footnote{A convex cone is said to be closed if it contains all it's limit points with respect to some topology. 
    In the present case, since we work with finite dimensional real vector spaces, we take the standard topology. 
    A convex cone $P$ is said to be pointed if $P \cap -P = \{ 0 \}$.},
    \begin{align*}
        \forall n \in \NN, \P^{(n)} \subseteq V^{\otimes n},
    \end{align*}
    where $V^{\otimes n}$ is the standard $n$-fold tensor product of $V$.
    \item A distinguished element $\1$ in the relative interior\footnote{The relative interior of a convex cone is the interior of the cone with the ambient space taken to be the affine hull of the cone.} of $\P^{(1)}$.
\end{itemize}
We call $\P^{(n)}$ the \emph{effect cones} and $\1$ the \emph{unit effect}.
The effects model measurement outcomes, in the sense that the effect $e$ encodes the probability of obtaining the outcome labelled by $e$.
The unit effect in particular, models the trivial measurement, which encodes the fact that in every experiment some outcome must occur.
Therefore, a set of effects $\{e_i\}_{i = 1}^m \subseteq \P^{(n)}$ defines a measurement if, 
\begin{align*}
    \sum_{i = 1}^m e_i = \1^{\otimes n}. 
\end{align*}

In this formalism, \emph{states} are viewed as linear functionals on effects.
To this end let $(\P^{(n)})' \subseteq (V^{\otimes n})^*$ be the polar dual\footnote{The polar dual $P'$ of a convex cone $P$ is the set of all non-negative linear functionals on $P$.} of $\P^{(n)}$.
We choose a set of pointed cones $\forall n \in \NN : \D^{(n)} \subseteq (\P^{(n)})'$, that are generating\footnote{A convex cone is said to be generating if it spans the ambient space.} on $\Span(\P^{(n)})^*$, called the \emph{state cones}.
The convex sets
\begin{align*}
    \S^{(n)} := \{ \rho \in \D^{(n)} \ | \ \rho(\1^{\otimes n}) = 1 \},
\end{align*}
are called the \emph{state spaces} of the GPT.
The state spaces are compact sets (App.~\ref{app:app_to_formalism}), a fact that will be important later (Sec.~\ref{subsec:prop_of_tel_semigrp}).

The first deviation from the formalism in \cite{dmello2024entanglement} is that we do not require that the state and effect cones are generating with respect to the ambient spaces $V^{\otimes n}$ and their duals.
This is in order to be able to discuss the notion of local tomography in the latter part of this work (Sec.~\ref{subsec:no_all_pancakes}).
More precisely, in this paper, we will look at cases where $\Span(\P^{(2)}) = V^{\otimes 2}$ but $\Span(\P^{(1)}) \subsetneq V$.
Further, it may seem peculiar that we demand that $\D^{(n)}$ be a generating cone on $\Span(\P^{(n)})^*$.
This is all to say that, on a mathematical level, all we are doing is embedding cones defined as per \cite{dmello2024entanglement,plavala2023general} into (possibly) larger ambient spaces. 

Let $e \in \P^{(n)}$, define the \emph{negation} of $e$ as 
\begin{align*}
    \neg e := \1^{\otimes n} - e.
\end{align*}
The (convex) set of all those elements of $\P^{(n)}$ whose negation is also in $\P^{(n)}$ is called the \emph{effect space} $\E^{(n)}$, i.e., 
\begin{align*}
    \E^{(n)} := \P^{(n)} \cap \neg \P^{(n)}.
\end{align*}
This definition implies that 
\begin{align*}
    e \in \E^{(n)} \Leftrightarrow e \in \P^{(n)} \ \text{and} \ \neg e \in \P^{(n)},
\end{align*}
from which it directly follows that 
\begin{align*}
    e \in \E^{(n)} \Leftrightarrow \forall \rho \in \S^{(n)}, 0 \leq \rho(e) \leq 1.
\end{align*}
That is, the pairing between any effect $e$ and state $\rho$ is not only positive, but also bounded above by $1$.
As such it will be interpreted as the probability of obtaining the outcome corresponding $e$ given we prepared the state $\rho$.
Additionally, every element of $\P^{(n)}$ with this property is in $\E^{(n)}$.


\subsubsection{Constructing a GPT from a set of states and effects}\label{subsubsec:generating_a_GPT}

It is often the case that we seek to determine whether it is \emph{possible} to achieve certain correlations in some experimental scenario.
Addressing this question typically only requires us to specify states and measurements which, if employed in the experiment, reproduce the correlations in question.
Therefore, it is convenient to have an algorithmic procedure that takes a collection of states and effects and turns it into a GPT that contains them.

Algorithm~1 of \cite{dmello2024entanglement} provides a straight-forward ``closure construction'' that addresses exactly this problem.
In this work, in contrast to the treatment in \cite{dmello2024entanglement}, we do no assume that GPTs are invariant under permutation of subsystems.
Thus, when we refer to the \emph{GPT closure} of sets of states and effects, we mean the output of Algorithm~1 of \cite{dmello2024entanglement} with the symmetrization step omitted.

The focus of this work is the iterated CHSH game introduced in \cite{dmello2024entanglement}, which is a generalization of the adaptive CHSH game of \cite{weilenmann2020self,weilenmann2020towards}.
It can be stated by combining only bipartite states and effects.
For this reason, we will typically consider only GPTs generated by taking the GPT closure of bipartite objects.


\subsection{The CHSH scenario}\label{subsec:CHSH_scenario}

The CHSH scenario consists of two space-like-separated parties (usually called Alice and Bob), each of whom possesses a two-setting two-outcome measurement machine.
These two parties perform an experiment where, in every round, they choose one of two measurement settings and measure a shared bipartite state.

Given a GPT, define \emph{an instance of the CHSH scenario} as a choice of four effects -- $e_i, f_j \in \E^{(1)},\ i,j \in \{0,1\}$ -- and a bipartite state $\rho \in \S^{(2)}$.
The effect $e_0$ specifies the two-outcome measurement corresponding to setting $0$ of Alice, i.e., the effects $e_0, \neg e_0$ model the two outcomes respectively.
Analogously $e_1$ for setting $1$ of Alice, and, $f_0$ and $f_1$ for Bob.
$\rho$ models the bipartite state shared by Alice and Bob.

Let $A_i, B_j,\ i,j \in \{0,1\}$ be the $\pm 1$ valued observables corresponding to the measurement settings of Alice and Bob respectively.
In terms of the effects they are defined as follows:
\begin{align}\label{eqn:correlators_from_effects}
    A_i := e_i - \neg e_i,
    \quad
    B_j := f_j - \neg f_j.
\end{align}
The observable of interest for the CHSH scenario is given by 
\begin{align*}
    \B :=  A_0 \otimes B_0 + A_0 \otimes B_1 + A_1 \otimes B_0 - A_1 \otimes B_1.
\end{align*}
We call this the \emph{standard CHSH observable}.

Given an instance -- $\rho, e_i, f_j$ -- of the CHSH scenario, the pairing $\rho(\B)$ is called the \emph{CHSH value of the instance}.
The \emph{CHSH value of the GPT} is then the supremum over the CHSH values of all possible instances from the GPT.

The \emph{CHSH inequality} is the following Bell inequality
\begin{align*}
    \rho(\B) \leq 2,
\end{align*} 
respected by all classical theories.
A GPT is said to \emph{violate CHSH} if there exist at least one instance that violates the CHSH inequality.


\subsection{The teleportation semigroup}\label{subsec:telp_semi_grp}

Here we provide a short summary of the formalization of teleportation in the GPT framework (Ref.~\cite{barnum2008teleportation, chiribella2017quantum}).
In analogy to QM, teleportation refers to the process where we perform a joint bipartite measurement on one half of a bipartite state and a local state, producing a local state as a result (as depicted in Fig.\ref{fig:teleportation-def}).
For the mathematical definition, it is useful to keep in mind that we here regard states and linear functionals on effects.
Given a bipartite state $\rho$, a bipartite measurement $\{ \phi_k \}_k$, and a local state $\sigma$, conditioned on obtaining the outcome $k$, we get the state $\sigma_k$ defined by the following map:
\begin{align*}
    \sigma_k : e \mapsto \langle \phi_k \otimes e, \sigma \otimes \rho \rangle,
\end{align*}
for all local effects $e$.

\begin{figure}[b]
    \centering
    \includegraphics[trim={0 3.5cm 0 2.5cm}, clip, height=25mm, keepaspectratio]{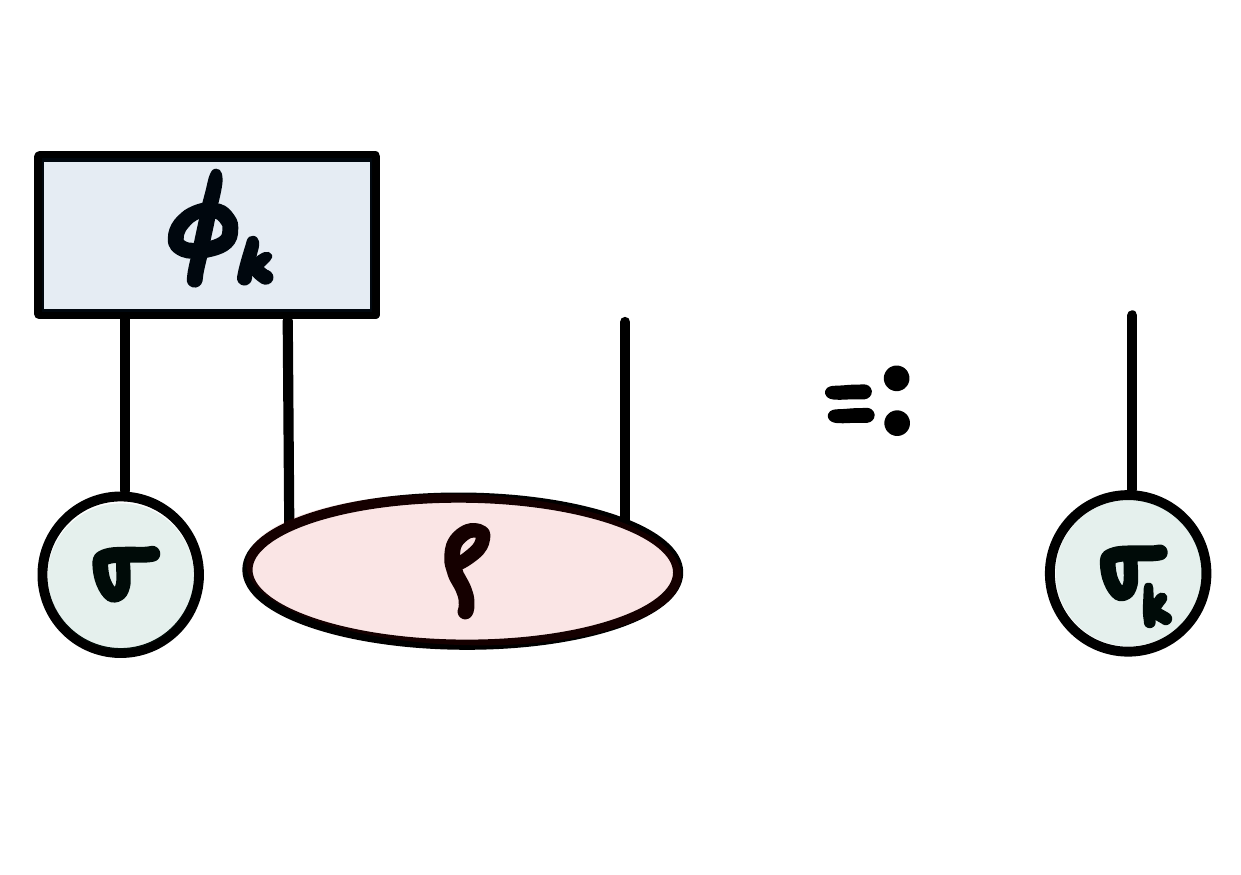}
    \caption{Diagrammatic representation of teleportation.
    The joint measurement $\{ \phi_k \}_k$ on a local state $\sigma$ and one-half of a bipartite state $\rho$, results in a local state $\sigma_k$, when conditioned on the measurement outcome corresponding to $\phi_k$.}
    \label{fig:teleportation-def}
\end{figure}

An equivalent way to look at the same situation is as follows:
Bipartite states are isomorphic to bilinear forms on local effects, and vice versa.
Thus with every bipartite state we can associate a map $\hat\rho$ from the effect cone to the state cone, defined by
\begin{align}\label{eqn:rhohat}
    \hat\rho : e \mapsto \rho(e, \cdot).
\end{align}
In other words, $\hat\rho$ sends every effect $e$ to the functional $\rho(e, \cdot)$.
Similarly, bipartite effects can also be viewed as maps from the state cone to effect cone 
\begin{align*}
    \hat\phi_k : \sigma \mapsto \phi_k(\sigma, \cdot).
\end{align*}

We use the above equivalence liberally in this work, i.e., we take $\rho$ and $\hat\rho$ to mean one and the same thing, operationally speaking.
Under this isomorphism pairings between bipartite states and bipartite effects become the trace of a product of maps.
There are two ways to do this, the usual pairing (Fig.~\ref{fig:pairings}~(a)) is given by
\begin{align*}
    \rho(\phi_k) = \tr(\hat\rho^T\hat\phi_k),
\end{align*}
where $\hat\rho^T$ is the map corresponding to the bilinear form obtained by switching the tensor factors of $\rho$ (upon choosing a basis, this corresponds to the matrix transpose).
It is also possible to stagger the states and effects and ``close the loop'' in order to pair them (Fig.~\ref{fig:pairings}~(b)).
In this case the pairing is given by $\tr(\hat\rho\hat\phi_k)$, and, depending on the situation, it might not have an analog to $\rho(\phi_k)$.

A special case of Fig.~\ref{fig:pairings}~(a) is when the effect of product type.
Given local effects, $e$ for Alice and $f$ for Bob, $\widehat{e \otimes f}$ denotes the map 
\begin{align*}
    \widehat{e \otimes f} : \sigma \mapsto \sigma(e) f.
\end{align*}
In terms of this map, we can write the pairing $\rho(e,f)$ as
\begin{align*}
    \rho(e,f) = \tr(\hat\rho^T \widehat{e \otimes f}) = \tr(\hat\rho\widehat{f \otimes e}) = \hat\rho(e)(f).
\end{align*} 
The last term in the above equation denotes the following: First apply the map $\hat\rho$ to the effect $e$ of Alice to obtain the state $\hat\rho(e)$ on Bob.
Then evaluate $\hat\rho(e)$ on the effect $f$ of Bob.

\begin{figure}[t]
    \centering
    (a)
    \includegraphics[trim={4cm 0 4.5cm 0}, clip, height=25mm, keepaspectratio]{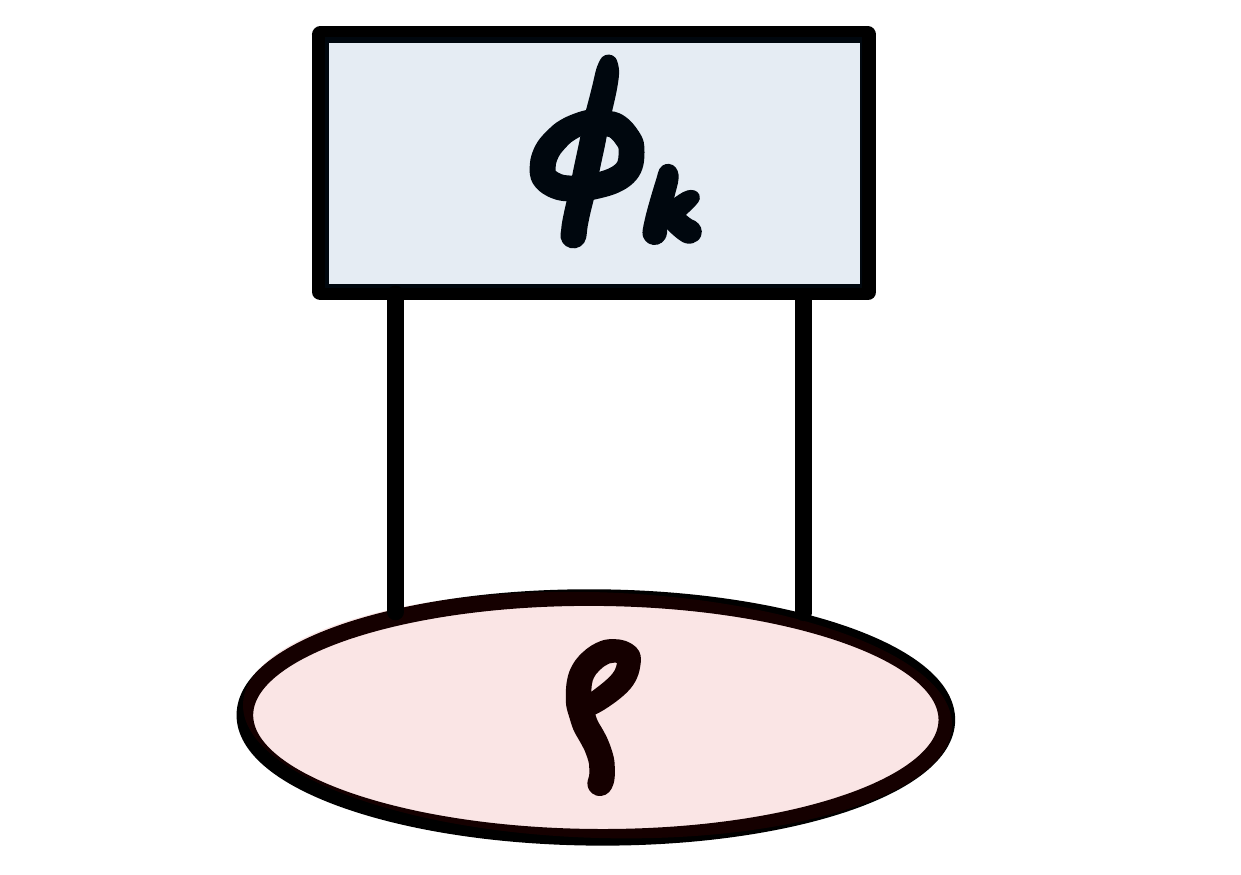}
    \hspace{2em}
    (b)
    \includegraphics[trim={0 1cm 0 1cm}, clip, height=25mm, keepaspectratio]{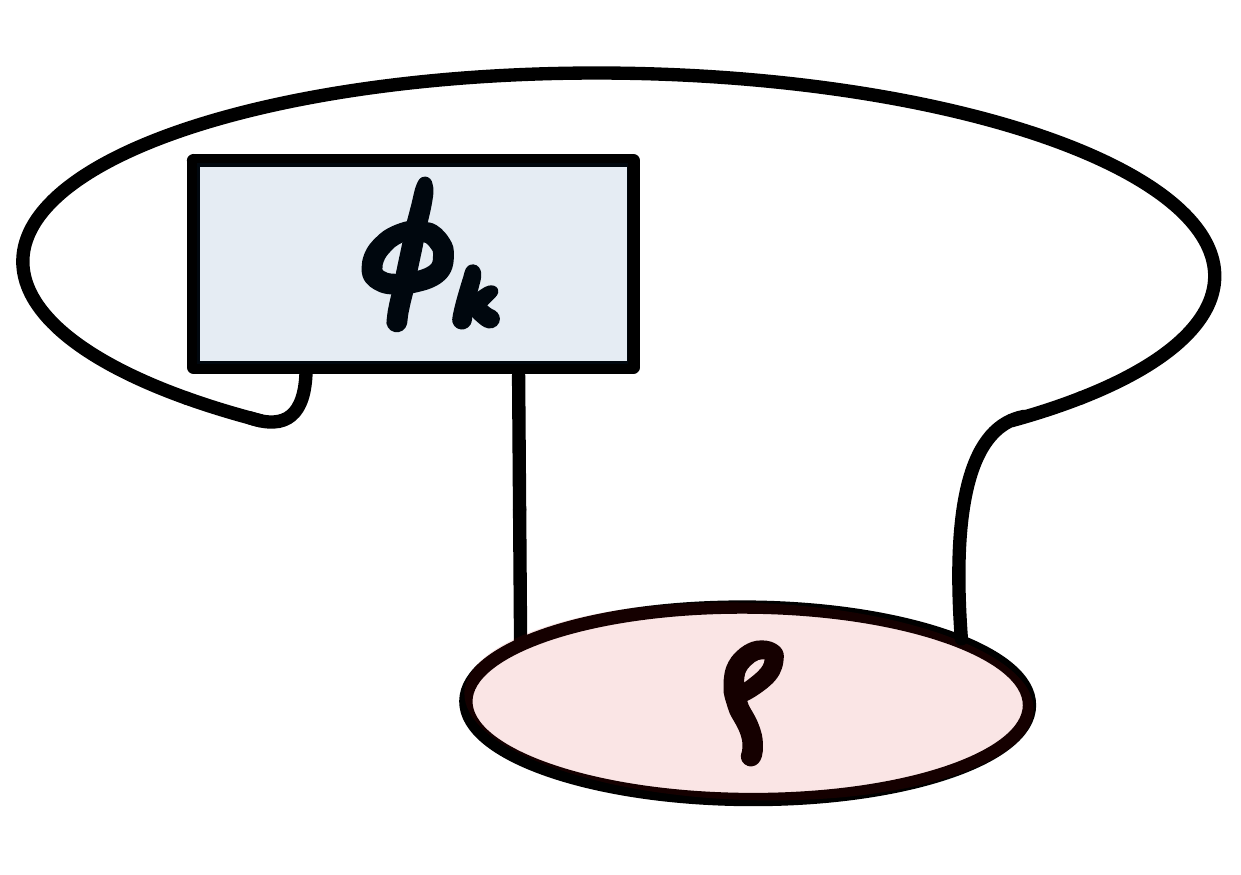}
    \caption{The possible ways to pair a bipartite state with a bipartite effect. 
    (a) Depicts the usual pairing $\rho(\phi_k)$. 
    In order to translate it into the trace pairing, we have to introduce a transpose to one of the two maps: $\tr(\hat\rho^T\hat\phi_k)$.
    (b) Depicts the other way to pair the two, namely by ``closing the loop''. 
    In order to translate this pairing into a trace we do not have to introduce a transpose: $\tr(\hat\rho\hat\phi_k)$.
    Such pairings arise when the type of subsystems, which the ``legs'' of the states and effects correspond to, are restricted (see Sec.~\ref{subsec:elim_dofs}).}
    \label{fig:pairings}
\end{figure}

Concatenating the maps associated with bipartite states and effects yields $\hat\rho  \hat\phi_k$, a map from states to (sub-normalized) states.
The set of maps $\{ \hat\rho\hat\phi_k \}_k$ generates a semigroup.
Consider an element of this semigroup, e.g., $(\hat\rho\hat\phi_{k_1})  \cdots  (\hat\rho\hat\phi_{k_N})$.
Physically, the image of a local state under this map corresponds to
teleporting the local state $N$ times and obtaining the list of outcomes $(k_1, \cdots, k_N)$.
Mathematically, the resulting local state is sub-normalized by a factor equal to probability of obtaining the outcomes $(k_1, \cdots, k_N)$.
Given a GPT, we call the set of all maps generated in this way the \emph{teleportation semigroup} of the GPT.
The teleportation semigroup also captures the phenomenon of \emph{entanglement swapping}, which can be described as the application of elements of the semigroup to one half of a bipartite state.


\subsection{Outline}

In terms of the terminology introduced in Sec.~\ref{subsec:GPTs}, \ref{subsec:CHSH_scenario}, and \ref{subsec:telp_semi_grp}, 
the conjecture of \cite{weilenmann2020self} stems from the observation that large CHSH values and the existence of teleportation protocols are at odds, in the sense that teleportation works to decrease the strength of CHSH correlations.
Indeed, there are many GPTs with very strong CHSH correlations where the set of states accessible after (iterated) teleportation can produce only classical correlations, e.g., boxworld \cite{popescu1994quantum, barrett2007information, plavala2023general}, composite GPTs \cite{dmello2024entanglement}, even-polygon theories \cite{janotta2011limits,weilenmann2020towards, weilenmann2020self}.
However, in our previous work \cite{dmello2024entanglement} we presented a first example of a GPT (OST) that retains a CHSH value of $4$, indefinitely, under entanglement swapping.

In view of these results, one can ask: What are the constraints imposed on GPTs by the requirement that they maintain a high CHSH value even after repeated teleportation?
Here we answer this question under mild technical assumptions. 
In Section~\ref{sec:GPT_self-testing} we develop a tool required to arrive at our results, namely GPT self-testing, which is a generalization of quantum self-testing to the GPT framework.
In Section~\ref{sec:classification}, we show that if certain correlations (c.f.\ Sec.~\ref{subsec:sufficient_cond_for_classification}) are maintained in the iterated CHSH scenario, then there exists an \emph{effective GPT} (c.f.\ Sec.~\ref{subsec:eff_I_CHSH_GPT}) on which the teleportation maps are reversible transformations, thus form a group. 
The positivity conditions on this group will turn out to be highly constraining (c.f.\ Sec.~\ref{subsec:5+2_families}):
Every such GPT is a member of one of seven, inequivalent, representation-theoretically defined families.
Towards the end of this section we discuss the quantum realizations of these families (Sec.~\ref{subsec:Q_family}), and the implications our classification results have for the \emph{local tomography} assumption (Sec.~\ref{subsec:no_all_pancakes}).

Our results are to be contrasted with the work of D'Ariano, Chiribella, and Perinotti \cite{chiribella2017quantum}.
They have discussed a system of axioms on GPTs that are strong enough to derive QM.
They also arrive at the conclusion that there must be a compact group structure associated with teleportation.
However, the axioms used in their work are much stronger than what we consider here.
For example, we do not assume that every state has a purification.


\section{GPT self-testing}\label{sec:GPT_self-testing}

Self-testing is a notion that is established in the setting of quantum theory.
In essence, it refers to the fact that:
\begin{center}
    ``For certain correlations, the quantum model realizing the correlation is unique up to isomorphisms''.
\end{center}
For example, in the CHSH scenario, if we measure a CHSH value of $2\sqrt2$ then any quantum model that explains this correlation is isomorphic to specific Pauli observables being measured on the Bell state \cite{summers1987bell,vsupic2020self}.

Here we generalize this notion to the framework of general probabilistic theories in the sense that:
\begin{center}
    ``For certain correlations, any two instances (from possibly different GPTs) achieving the correlation lead to equivalent effective GPTs''.
\end{center}
Here, by \emph{effective GPT} we mean ``the GPT obtained by discarding irrelevant degrees of freedom''.
In the case of CHSH, given an instance of the CHSH scenario, we can define \emph{the effective CHSH GPT} as ``the GPT obtained by ignoring all degrees of freedom unreachable from the instance''. 
Concretely, 
\begin{definition}[The effective CHSH GPT]
    Let $\rho, e_i, f_j$ be an instance of the CHSH scenario.
    Define $\Omega_A := \{e_i, \neg e_i\}_{i = 0,1}$, and $\Omega_B := \{f_j, \neg f_j\}_{j = 0,1}$.
    Then, the \emph{effective CHSH GPT} is the GPT closure (c.f.\ Sec.~\ref{subsubsec:generating_a_GPT}) of the input: 
    $P = \cone\big(\Omega_A \otimes \Omega_B\big)$, $D = \cone\big( \{\rho\} \cup \rho(\cdot, \Omega_B)\otimes\rho(\Omega_A, \cdot) \big)$.
\end{definition}

As mentioned in Sec.~\ref{subsubsec:generating_a_GPT}, in the above definition the permutation symmetry between the subsystems is broken.
The left subsystem of $\rho$ must always be paired with Alice's effects and the right subsystem with Bob's.

The notion of two GPTs being equivalent to each other is formalized as follows:
Given that a GPT is specified by the following tuple of data
\begin{align*}
    \Big(V, \1, \{\P^{(n)}\}_{n \in \NN}, \{\D^{(n)}\}_{n \in \NN}\Big),
\end{align*}
we can define:
\begin{definition}[GPT isomorphism]
    Two GPTs
    \begin{gather*}
        \Big(V_1, \1_1, \{\P_1^{(n)}\}_{n \in \NN}, \{\D_1^{(n)}\}_{n \in \NN}\Big)
        \\
        \text{and}
        \\
        \Big(V_2, \1_2, \{\P_2^{(n)}\}_{n \in \NN}, \{\D_2^{(n)}\}_{n \in \NN}\Big)
    \end{gather*}
    are said to be isomorphic if there exists a linear map $\gamma : V_1 \to V_2$ such that:
    \begin{enumerate}
        \item $\gamma(\1_1) = \1_2$;
        \item For every $n \in \NN$, the restriction 
        \begin{align*}
            \gamma^{\otimes n} : \Span(\P_1^{(n)}) \to \Span(\P_2^{(n)})
        \end{align*}
        is invertible, with 
        \begin{align*}
            \gamma^{\otimes n}\P_1^{(n)} = \P_2^{(n)}
            \quad 
            \text{and}
            \quad 
            (\gamma^{-t})^{\otimes n}\D_1^{(n)} = \D_2^{(n)},
        \end{align*}
        where $\gamma^t$ is the canonical transpose\footnote{Given a linear map $L: V \to W$, its \emph{canonical transpose} $L^t : W^* \to V^*$ is defined by: $\forall f \in W^*, v \in V,\ (L^t f)(v) = f(Lv)$.} of $\gamma$.
    \end{enumerate}
\end{definition}


\subsection{CHSH self-testing correlations}\label{subsec:self-test_corr}

In quantum theory, any instance that achieves a CHSH value of $2\sqrt 2$ leads to a unique effective CHSH GPT.
This is a consequence of the well-known stronger result that even the quantum model is unique \cite{summers1987bell,vsupic2020self}.
The main result of this section is that any theory achieving a CHSH value of $4$ also has a unique effective CHSH GPT.

First we show that the following pattern of correlations are sufficient to self-test the corresponding GPTs:

\begin{lemma}\label{lem:sufficient_st_conditions}
    Given an instance -- $\rho, e_i, f_j$ -- of the CHSH scenario, the following conditions are sufficient for uniqueness of the effective CHSH GPT: For $A_i, B_j$ as defined in Eqn.~\ref{eqn:correlators_from_effects}, and some $a \in (\tfrac12, 1]$,
    \begin{enumerate}
        \item $\rho(A_i \otimes \1) = \rho(\1 \otimes B_j) = 0$;
        \item $\rho( (-1)^{ij} A_i \otimes B_j) = a$.
    \end{enumerate}
\end{lemma}

The first condition translates to the marginals being uniform.
The second conditions encodes both that all four correlators contribute equally and that the CHSH inequality is violated.

\begin{proof}
    We prove this by showing that the above conditions fix every pairing needed to specify the effective CHSH GPT.
    Using $A_i = e_i - \neg e_i$, condition $1$ yields 
    \begin{align*}
        \rho(e_i \otimes \1) 
        &
        = \rho(\neg e_i \otimes \1)
        \\
        &
        = \rho(\1 \otimes \1) - \rho(e_i \otimes \1)
        \\
        &
        = 1 - \rho(e_i \otimes \1).
    \end{align*}
    Applying the same argument to $B_j$ we can conclude that 
    \begin{align}\label{eqn:sufficient_1}
        \rho(e_i \otimes \1) = \rho(\1 \otimes f_j) = \tfrac12.
    \end{align}

    Now consider $\rho(A_i \otimes B_j)$.
    Using Eqn.~(\ref{eqn:sufficient_1}), and $v - \neg v = 2v - \1$, we can re-write this as 
    \begin{align*}
        &
        \rho(A_i \otimes B_j) 
        \\
        &
        = 4\rho(e_i \otimes f_j) - 2 (\rho(e_i \otimes \1) + \rho(\1 \otimes f_j)) + \rho(\1 \otimes \1)
        \\
        &
        = 4\rho(e_i \otimes f_j) - 1.
    \end{align*}
    Thus from condition $2$ we get 
    \begin{align*}
        \rho(e_i \otimes f_j) = \tfrac14(1 + (-1)^{ij} a).
    \end{align*}

    Finally, $a>\tfrac12$ implies that the CHSH inequality is violated.
    As a consequence, $\{\1, e_0, e_1\}$ must be linearly independent.
    This is because, if it were not the case, there would exist a joint measurement machine for the measurements $\{e_0, \neg e_0\}$ and $\{e_1, \neg e_1\}$, implying that the CHSH inequality cannot be violated.
    Similarly for $\{\1, f_0, f_1\}$.
    Thus if there is another instance $\rho', e_i', f_j'$ with the same CHSH value (i.e., $4a$) satisfying the two conditions, then the map specified by $e_i \mapsto e_i'$, $f_j \mapsto f_j'$, and $\1 \mapsto \1'$ is a GPT isomorphism.
\end{proof}

Using the above Lemma, we can conclude the following:

\begin{lemma}
    Any two instances achieving a CHSH value of $4$ lead to equivalent effective GPTs.
\end{lemma}


\begin{proof}
    Let -- $e_i, f_j, \rho$ -- be an instance of the CHSH scenario with CHSH value 4.
    Since $|\rho(A_i \otimes B_j)| \leq 1$ we necessarily have 
    \begin{align}\label{eqn:4.corr}
        \rho((-1)^{ij} A_i \otimes B_j) = 1.
    \end{align}
    Applying $\rho$ to
    \begin{align*}
			A_i\otimes B_j &= e_i \otimes f_j + \neg e_i \otimes \neg f_j - (e_i \otimes \neg f_j + \neg e_i \otimes f_j), \\
			\1\otimes \1   &= e_i \otimes f_j + \neg e_i \otimes \neg f_j + e_i \otimes \neg f_j + \neg e_i \otimes f_j 
    \end{align*}
    directly yields
    \begin{align*}
        \rho(e_i \otimes \neg f_j) = \rho(\neg e_i \otimes f_j) & = 0 & (i,j)\neq(1,1), \\
        \rho(e_1 \otimes f_1) = \rho(\neg e_1 \otimes \neg f_1) & = 0.
    \end{align*}
    Plugging in the definition of negation for the case $(i,j) \neq (1,1)$ gives 
    \begin{align}
        &&
        \rho(e_i\otimes \1) - \rho(e_i \otimes f_j) 
        &=
        \rho(\1 \otimes f_j) - \rho(e_i \otimes f_j) \nonumber \\
        &\Rightarrow&
        \rho(e_i\otimes \1)
        &=
        \rho(\1 \otimes f_j).
        \label{eqn:4.1}
    \end{align}
	For the case $(i,j) = (1,1)$, it gives
    \begin{align}
        &&
        \rho(e_1\otimes f_1) &= 1 - \rho(\1\otimes f_1) - \rho(e_1\otimes \1) + \rho(e_1\otimes f_1) \nonumber \\
        &\Rightarrow&
        1 &=  \rho(\1\otimes f_1) + \rho(e_1\otimes \1).
        \label{eqn:4.2}
    \end{align}
    Combining Eqs.~(\ref{eqn:4.1})~and~(\ref{eqn:4.2}) yields 
    \begin{align*}
        \rho(e_i\otimes \1) = \rho(\1\otimes f_j) = \tfrac 12.
    \end{align*}
    This is equivalent to the uniform marginals condition (see proof of Lem.~\ref{lem:sufficient_st_conditions}).
\end{proof}


\section{Classification of Probabilistic theories with a stable CHSH value under teleportation.}\label{sec:classification}

Here we investigate the structure imposed on GPTs by the requirement that they maintain their CHSH value even after multiple rounds of teleportation (entanglement swapping). 
The clear choice of task for this purpose is the iterated CHSH game introduced in \cite{dmello2024entanglement}. 
For a number of rounds $N \in \NN$, the iterated CHSH game is played by $N+2$ parties -- Alice, Charlie and $N$ further parties, whom we will refer to as the ``Bobs''.
The game is specified by the following data: A bipartite state $\rho \in \S^{(2)}$, an $n$-outcome bipartite measurement $\M = \{\phi_k\}_{k \in [n]} \subseteq \E^{(2)}$ ($[n]$ denotes the index set $\{1, \ldots, n\}$), and two pairs of two-outcome measurements $\Omega_A = \{ e_i, \neg e_i \}_{i = 0,1} \subseteq \E^{(1)}$ and $\Omega_C = \{ f_j, \neg f_j \}_{j = 0,1} \subseteq \E^{(1)}$.
These data remain fixed throughout the game.
In particular, this implies that Alice and Charlie are not allowed to change their measurement devices.
Each round of the game proceeds as follows:
\begin{enumerate}
    \item 
        First, the $N$ Bobs perform the bipartite measurement $\M$ on $N+1$ copies of the bipartite state $\rho$ to produce: (i) The bipartite state $\sigma$ (as depicted in Fig.~\ref{fig:i-ent-swap}~(a)), and (ii) A list of outcomes $\vec k \in [n]^N$;
    \item 
		Next Alice and Charlie perform a CHSH test on $\sigma$.
        Given the outcome $\vec k$ they are allowed to declare which version CHSH inequality they intend to test.
        The eight different versions are related to each other by a relabelling of settings, outcomes, and parties (see App.~\ref{app:action_on_chsh_obs}).
\end{enumerate}

\begin{figure}[t]
    \centering
    \includegraphics[trim={0 6cm 0 5cm}, clip, width=\linewidth, keepaspectratio]{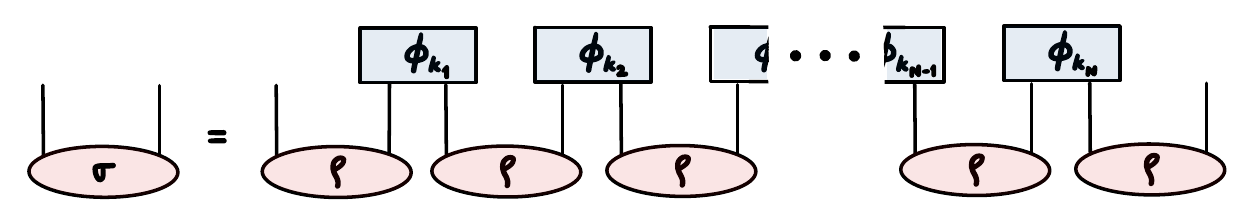}
    (a)
    
    \vspace{10pt}

    \includegraphics[trim={0 3cm 0 3cm}, clip, width=\linewidth, keepaspectratio]{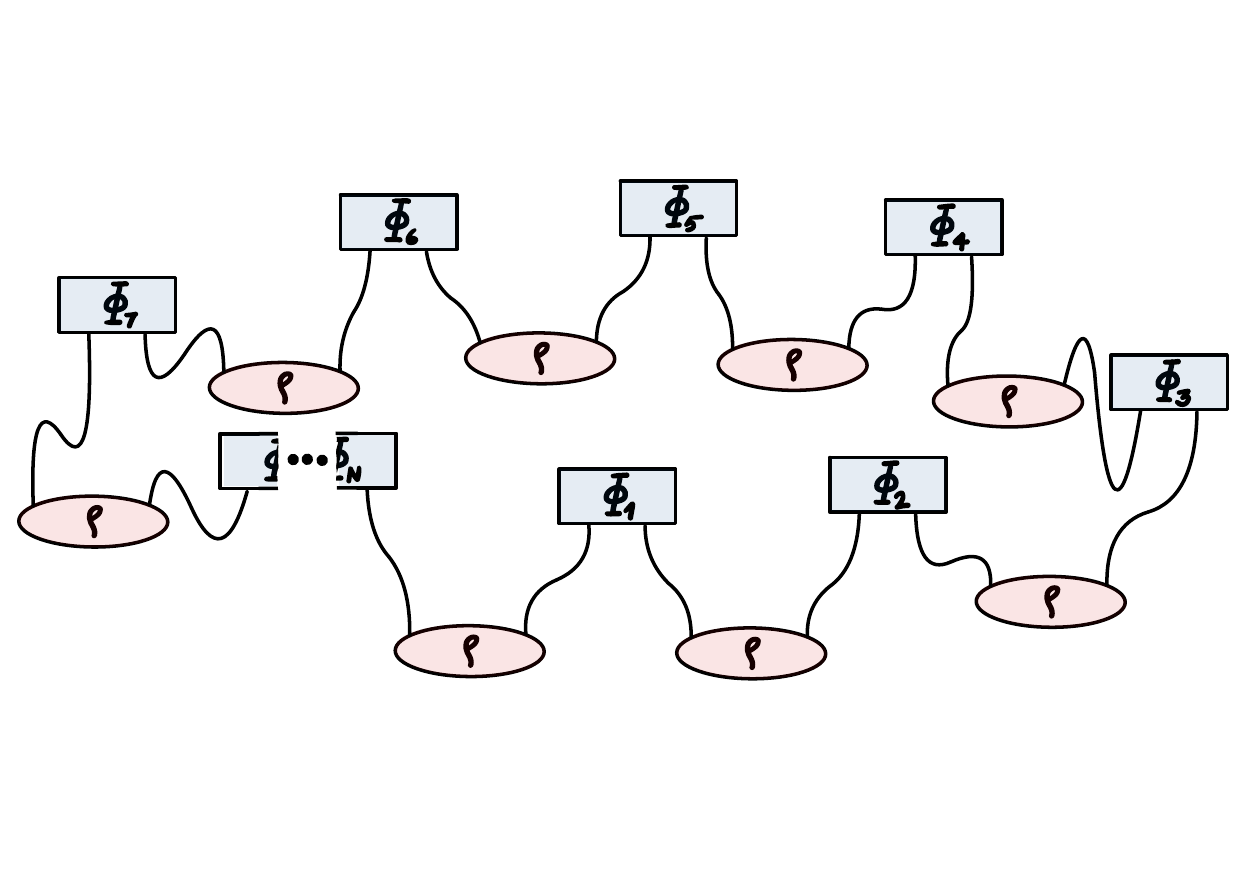}
    (b)
    \caption{\textbf{(a)} The resulting (sub-normalized) state after the Bobs perform entanglement swapping.
    \textbf{(b)} The most general experiment that can be performed given the data specifying the iterated CHSH game.
    Here $\Phi_i \in \conv(\M \cup \Omega_C \otimes \Omega_A)$, i.e., that can be either product or entangled (or a convex combination of both).
    The iterated CHSH game is a special case of this setup where e.g.\ $\Phi_1$ is chosen to be a product effect from $\Omega_C \otimes \Omega_A$ and the rest of the $\Phi_i$ are chosen from $\M$.
    This is because choosing one of the $\Phi_i$ to be a product effect allows us to ``break the loop'' and unravel it into an instance of the Iterated CHSH game.}
    \label{fig:i-ent-swap}
\end{figure}

The goal of this section is to construct the effective GPT, analogously to one introduced in Sec.~\ref{sec:GPT_self-testing}, describing every experiment that can be built out of the data specifying the iterated CHSH game.
One might assume that this includes only experiments of the kind where Alice and Charlie perform a (product) measurement on the type of states depicted in Fig.~\ref{fig:i-ent-swap}~(a).
This is not the case since it is also possible to bring the ends of this long chain together and measure $\M$ on it again.
Because such a situation must also be described by our GPT, the most general scenario, given an instance of iterated CHSH, is the one depicted in Fig.~\ref{fig:i-ent-swap}~(b).
In this figure, the effects $\Phi_i$ can be taken to be either product or entangled.
To capture the fact that it is always possible to realize correlations that are in the convex hull of these two situations, we take $\Phi_i \in \conv(\M \cup \Omega_C \otimes \Omega_A)$.


\subsection{Eliminating irrelevant degrees of freedom}\label{subsec:elim_dofs}

In Sec.~\ref{subsec:GPTs}, we introduced the space $V$ in terms of which the mathematical descriptions of measurements and states are modeled.
The effective GPT will contain only a subset of all states and measurements.
As a result, the original spaces, $V$ and $V^*$, are ``too large'', e.g.\ in the sense that $V^*$ now contains mathematically distinct states that are no longer distinguishable with respect to the remaining measurements.
Therefore, in the spirit of self-testing we eliminate the now redundant degrees of freedom from our mathematical description.

As one can see from Fig.~\ref{fig:i-ent-swap}~(b), Alice will always perform measurements on the left subsystem of the state $\rho$, while Charlie only has access to the right subsystem.
Likewise, the left subsystem of a bipartite effect in $\{\phi_k\}_{k \in [n]}$ is contracted with the right subsystem of a state $\rho$ and vice-versa.
Thus the invariance under permutation of subsystems is broken.
We will refer to the left subsystem of $\rho$ as being of ``$A$-type'', and the right subsystem as being of ``$C$-type''.
Instead of using one mathematical object $V$ for all single-party subsystems, a minimal description will assign different spaces $V_A, V_C$ to these two different types.

\begin{lemma}
	The effective GPT can be described in terms of spaces $V_A, V_C$ with the following property:
	The map $\hat\rho: V_A\to V_C^*$,
	associated to the bipartite state $\rho$ in the sense of Eq.~(\ref{eqn:rhohat}),
	is a linear isomorphism.
\end{lemma}

\begin{proof}
    For a number $N \in \NN$ of rounds, the probability distributions generated by experiments of the form Fig.~\ref{fig:i-ent-swap}~(b) are given by
	\begin{align*}
			\tr\Big( (\hat\rho\hat\Phi_1) \cdots (\hat\rho\hat\Phi_N) \Big),
	\end{align*}
	where $\Phi_i \in \conv\big(\M \cup (\Omega_C \otimes \Omega_A)\big)$.

	Now, Let $\ker \hat\rho$ and $\img \hat\rho$ be the kernel and image of the map $\hat\rho$ respectively.
	Choose decompositions 
	\begin{align*}
			V = U \oplus \ker \hat\rho,
			\qquad
			V^* = \img \hat\rho \oplus W^*.
	\end{align*}
	Clearly $\hat\rho$ induces an isomorphism $\tilde\rho : U \to \img\hat\rho$.
    Thus we claim that 
    \begin{align*}
		V_A := U, \qquad V_C := (\img\hat\rho)^*,
	\end{align*}
    are the spaces whose existence is posited by the lemma.

    To prove this, we have to show that, for an appropriate restriction of the effects, the probability distributions generated by the restricted state and effects are identical to the distributions generated by the original ones.
	Let $\pi$ be the projector onto $U$ along $\ker \hat\rho$, and let $\iota$ be the embedding of $\img \hat\rho$ into $V^*$.
	Then we can obtain the maps induced by the original effects as
	\begin{gather*}
		\tilde \phi_k := \pi \hat\phi_k \iota : \img\hat\rho \to U,
	\end{gather*}
    and similarly
    \begin{align*}
        \widetilde{f \otimes e} := \pi (\widehat{f \otimes e}) \iota,
    \end{align*}
    for $f \otimes e \in \Omega_C \otimes \Omega_A$.
	And indeed, this restriction leaves the probability distributions unchanged:
	\begin{align*}
			\tr\Big( (\tilde\rho\tilde\Phi_1) \cdots (\tilde\rho\tilde\Phi_N) \Big)
			&
			= \tr\Big( \tilde\rho(\pi\hat\Phi_1\iota) \cdots \tilde\rho(\pi\hat\Phi_N\iota) \Big)
			\\
			&
			= \tr\Big( (\iota\tilde\rho\pi)\hat\Phi_1 \cdots (\iota\tilde\rho\pi)\hat\Phi_N \Big)
			\\
			&
			= \tr\Big( (\hat\rho\hat\Phi_1) \cdots (\hat\rho\hat\Phi_N) \Big),
	\end{align*}
	where we have used that $\iota \tilde\rho \pi = \hat\rho$ by construction.
\end{proof}

Hereafter we assume that the above reduction has already been made, i.e., the data for the iterated CHSH game is: $\rho \in (V_A \otimes V_C)^*$, $\M \subseteq  V_C \otimes V_A$ and $\Omega_C \otimes \Omega_A \subseteq V_C \otimes V_A$.


\subsection{The relabelling group}

The iterated CHSH game allows Alice and Charlie to choose a CHSH inequality based on the measurement outcomes obtained by the Bobs.
All eight versions of the CHSH inequality arise from each other by exchanging the labels of the measurement settings or outcomes of just one of the two parties (see App.~\ref{app:action_on_chsh_obs}). 
It turns out that analyzing the group action of these ``relabelling operations'' on the mathematical description of the GPT is very fruitful.

As an example, say we relabel the settings of Alice.
Mathematically, this constitutes a group action on $\Omega_A$, namely $e_0 \leftrightarrow e_1$, $\neg e_0 \leftrightarrow \neg e_1$.
Operationally, this gives us another two-setting two-outcome measurement machine which violates the CHSH inequality with $A_0$ and $A_1$ exchanged, which links it to the following group action on the CHSH observables:
\begin{gather*}
    A_0C_0
    + A_0C_1
    + A_1C_0
    - A_1C_1
    \\
    \updownarrow
    \\
    A_0C_0
    - A_0C_1
    + A_1C_0
    + A_1C_1.
\end{gather*}
The outcome relabelling admits a similar description.
This relabelling group action has the following properties:
\begin{itemize}
    \item Exchanging setting and outcome labels each constitute a $\ZZ_2$ action.
    Since the exchange of setting labels \emph{permutes} the exchange of outcome labels, the full group is a wreath product (see \cite{Robinson1996}~Section~1.6) of the two $\ZZ_2$ actions.
	The wreath product of $\ZZ_2$ with itself is isomorphic to the dihedral group of order 8, denoted by $D_4$ (see exercises of \cite{Robinson1996}~Section~1.6).
    
    \item This group action is \emph{free} (i.e.\ every CHSH observable is stabilized only by the identity), and
			\emph{transitive} (i.e.\ any two CHSH observables can be mapped onto each other by a suitable relabeling).
			See App.~\ref{app:action_on_chsh_obs} for further details.
\end{itemize}
These properties imply that we can rephrase the second step of the iterated CHSH game as follows 
\begin{center}
    ``Based on $\vec k$, Alice and Charlie are allowed to pick the CHSH inequality that they will test.''

    $\updownarrow$

    ``Based on $\vec k$, Alice relabels her measurement device. Alice and Charlie then perform the CHSH test with respect to the standard CHSH observable.''
\end{center}

This rephrasing will allow us to relate the \emph{teleportation maps} (refer Fig.~\ref{fig:teleportation_map}) that appear in the entanglement swapping step of the iterated CHSH game to the local corrections applied by Alice, which forms a group.
Indeed, define the teleportation maps as follows:
\begin{align*}
	R_k &:= \hat\rho\hat\phi_k : V_C^*\to V_C^*.
\end{align*}
\begin{figure}[b]
    \centering
    \includegraphics[trim={0 0 0 0}, clip, height=25mm, keepaspectratio]{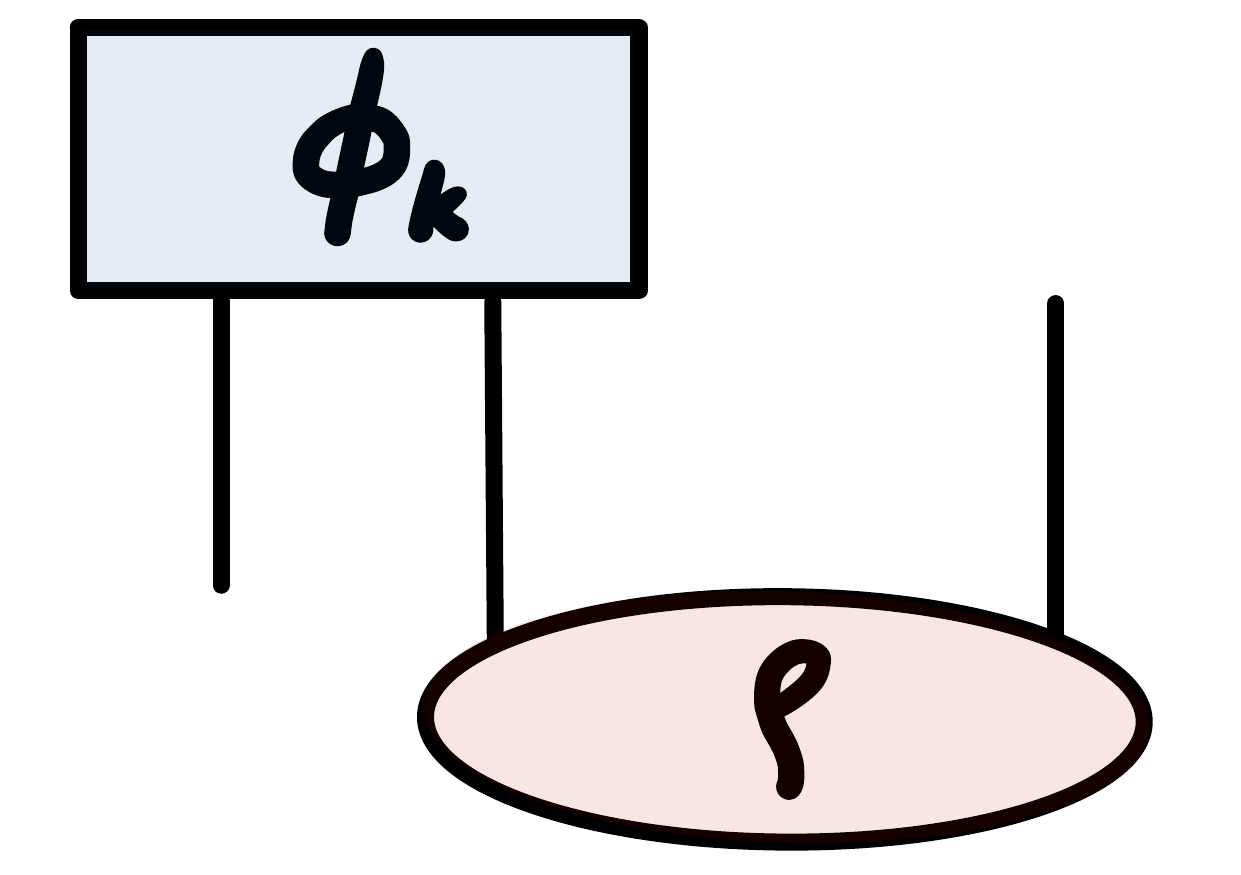}
    \caption{Diagrammatic representation of the teleportation map $R_k$.}
    \label{fig:teleportation_map}
\end{figure}
The corresponding state after after one round of entanglement swapping, conditioned on obtaining outcome $k$ of $\M$, can be written in terms of $R_k$ as:
\begin{align*}
    \hat\sigma = \tfrac1{p_k}R_k \hat\rho,
\end{align*} 
where $p_k := R_k  \hat\rho(\1)(\1)$, is the probability of obtaining the outcome $k$. 
Extend this definition to $N$ rounds as follows: For each outcome $\vec k = (k_1, \cdots, k_N)$ that can occur, define 
\begin{align*}
    R_{\vec k} := R_{k_1}  \cdots  R_{k_N}.
\end{align*}
Then the resulting bipartite state is 
\begin{align*}
    \hat\sigma = \tfrac1{p_{\vec k}} R_{\vec k}  \hat\rho,
\end{align*}
with $p_{\vec k}$ defined analogously.
Because $\hat\rho$ is an isomorphism, the state $\hat\sigma$ is in a one to one correspondence with the map $\tfrac1{p_{\vec k}}R_{\vec k}$.
By assumption, there exists a map that that assigns to each outcome $\vec k$ a correction from $D_4$.
The correction need only depend on the state realized between Alice and Charlie conditioned on obtaining $\vec k$, i.e., if two different outcomes $\vec k, \vec k'$ correspond to the same state, then the corrections must also be the same.
Hence, there exists a map $\vartheta$ that sends each realizable $\tfrac1{p_{\vec k}}R_{\vec k}$ to an element of $D_4$.
Let $S$ be the set of all realizable teleportation maps:
\begin{align}\label{eqn:def_of_S}
    S := \Big\{ \tfrac1{p_{\vec k}}R_{\vec k} \ \Big| \ \vec k \in [n]^N, N \in \NN \Big\}.
\end{align}
Then 
\begin{equation}\label{eqn:def_of_vartheta}
    \begin{aligned}
        \vartheta : \quad S \quad &\to D_4, \\
        \tfrac1{p_{\vec k}}R_{\vec k} &\mapsto \vartheta_{\vec k}.
    \end{aligned}
\end{equation}
Then, under the aforementioned rephrasing, we can describe the iterated CHSH game mathematically as follows: The Bobs each measure $\M$ and obtain outcome $\vec k$, which implies that the bipartite state between Alice and Charlie is $\tfrac1{p_{\vec k}}R_{\vec k}  \hat\rho$.
Next, Alice applies the relabelling $\forall e \in \Omega_A,\ e \mapsto \vartheta_{\vec k}^{-1} e$.
Finally, Alice and Charlie perform the usual CHSH test.
The resulting correlations are described by the distribution on $e \otimes f \in \Omega_A \otimes \Omega_C$ given by:
\begin{align*}
    e \otimes f \quad \mapsto \quad \tfrac1{p_{\vec k}} R_{\vec k}  \hat\rho(\vartheta_{\vec k}^{-1} e)(f).
\end{align*}

\emph{Remark:} In this formalization, $R_k$ acts on $C$-type subsystem of the state in Schr\"odinger picture, whereas $\vartheta_{\vec k}$ acts on $A$-type local observables in Heisenberg picture.
In particular, the ``last error introduced on Charlie's side, must be corrected first by Alice's correction'', which is why we chose the correction to be phrased in terms of $\vartheta_{\vec k}^{-1}$.
Below we'll see that these conventions make $\vartheta$ a homomorphism.


\subsection{Sufficient condition for classification}\label{subsec:sufficient_cond_for_classification}

The iterated CHSH scenario stands almost in analogy to the CHSH scenario in the following sense: There exists a sufficient condition on the correlations generated by an instance of the iterated CHSH game, satisfying which, the effective GPT generated will belong to one of seven inequivalent families (as opposed to being unique).

In this section, we state a sufficient condition for our classification result and derive some primliminary consequences of the same.
The fact that this condition is sufficient will only be clear by the end of Sec.~\ref{subsec:5+2_families}.
The condition is as follows:

\begin{condition}\label{cond:self-test_all_N}
    For the instance of the iterated CHSH game under consideration, it holds that, for every $N\in \NN$ and every outcome $\vec k$,
	\begin{itemize}
		\item
			the CHSH value remains constant, and is equal to the value without any entanglement swapping ($N = 0$), and
		\item
			it satisfies the two self-testing conditions of Lem.~\ref{lem:sufficient_st_conditions}.
	\end{itemize}
\end{condition}

\emph{Remarks:} (1) Condition~\ref{cond:self-test_all_N} allows for the possibility that the CHSH value of the instance is smaller than that of the GPT from which it originated.
As long as this CHSH value is preserved in the iterated CHSH game, the effective GPT (which will be constructed later) will be subject to the classification result.
(2) If the CHSH value of the theory is equal to $4$, then stability under teleportation implies Condition~\ref{cond:self-test_all_N}.

Now, without loss of generality, we can assume that every outcome $k \in [n]$ of $\M$ has non-zero probability ($p_k \neq 0$).
This is because every $\phi_k$ with $p_k = 0$ can be grouped (by summing the effects) with any $\phi_l$ with $p_l \neq 0$.
Under this assumption we have the following consequence of Condition~\ref{cond:self-test_all_N}:

\begin{lemma}\label{lem:correction_map}
	For any instance of the iterated CHSH game satisfying Condition~\ref{cond:self-test_all_N}, it holds that $S$ (Eqn.~(\ref{eqn:def_of_S})) is a semigroup and $\vartheta$ (Eqn.~(\ref{eqn:def_of_vartheta})) a semigroup homomorphism.
\end{lemma}

\begin{proof}
	Given Condition~\ref{cond:self-test_all_N}, Lemma~\ref{lem:sufficient_st_conditions} applies and we can conclude that 
	for every outcome $\vec k$, and local effects $e \in \Omega_A$ and $f \in \Omega_C$, 
	\begin{align*}
		\tfrac1{p_{\vec k}} R_{\vec k}  \hat\rho(\vartheta_{\vec k}^{-1} e)(f) =  \hat\rho(e)(f),
	\end{align*}
	i.e.\ the correlations after relabeling are identical to the ones obtained in the CHSH test without any entanglement swapping.

	Specialize the above to $N = 1$, and set $e' := \vartheta_{k}^{-1}e$, to get
	\begin{align*}
			\tfrac1{p_{k}}R_{k} \hat\rho(e')(f) = \hat\rho(\vartheta_{k}e')(f).
	\end{align*}
	Applying $\tfrac1{p_{l}}R_l$ to the left hand side and substituting $e' \mapsto e$ gives
	\begin{align*}
			\tfrac1{p_{l}}R_l\tfrac1{p_{k}}R_{k} \hat\rho(e)(f)
			&
			= \tfrac{1}{p_{l}p_{k}}R_{l \circ k} \hat\rho(e)(f)
			\\
			&
			= \tfrac{p_{l \circ k}}{p_{l}p_{k}} \hat\rho(\vartheta_{l \circ k}e)(f),
	\end{align*}
	while on the right hand side it gives 
	\begin{align*}
			\tfrac1{p_{l}}R_{l} \hat\rho(\vartheta_{k}e)(f)
			&
			= \hat\rho(\vartheta_{l}\vartheta_{k}e)(f).
	\end{align*}
	Equating both sides leads to 
	\begin{align*}
			\tfrac{p_{l \circ k}}{p_{l}p_{k}} \hat\rho(\vartheta_{l \circ k}e)(f) = \hat\rho(\vartheta_{l}\vartheta_{k}e)(f).
	\end{align*}
	Since the action of $\vartheta_{\vec k}$ is free, and $\forall k \in [n], p_k \neq 0$, we can conclude: $p_{l \circ k} = p_l p_k \neq 0$, and $\vartheta_{l \circ k} = \vartheta_l \vartheta_k$.
	Iterating this proves the claim.
\end{proof}

In light of Lem.~\ref{lem:correction_map} we can make two further assumptions about the measurement $\M$.
These assumptions are similar to the one made before Lem.~\ref{lem:correction_map}, i.e., they involve ``lowering the resolution'' of the measurement $\M$, without loss of generality, in order to make it simpler.
Thus we call this step ``coarse-graining''.

\begin{lemma}[Coarse-graining]\label{lem:coarse-graining}
    Without loss of generality we can assume the measurement $\M = \{ \phi_k \}_{k \in [n]}$ to have the following properties: 
    \begin{enumerate}
        \item Every correction in the image of $\vartheta$ can already be realized in $\M$, i.e., 
        \begin{align*}
            H := \{ \vartheta_k \}_{k \in [n]}
        \end{align*}
        is a subgroup of $D_4$;
        \item Every outcome $k \in [n]$ corresponds to a different correction;
    \end{enumerate}
\end{lemma}

\begin{proof}
    For (1), since the $D_4$ group is of order $8$, the set 
    \begin{align*}
        \bigcup_{i = 1}^8\{ \vartheta_{\vec k} \ | \ \vec k \in [n]^i \}
    \end{align*}
    contains all the corrections that could occur.
    Thus we replace $\M$ by the measurement
    \begin{align*}
        \tfrac18 \bigcup_{i = 1}^8 \{ \hat\rho^{-1}R_{\vec k} \ | \ \vec k \in [n]^i \}.
    \end{align*}
    Now group together every effect that corresponds to the same correction to achieve (2).
\end{proof}


\subsection{Properties of the teleportation semigroup}\label{subsec:prop_of_tel_semigrp}

In this section we investigate the properties of the semigroup $S$ (Lem.~\ref{lem:correction_map}) generated by an instance of the iterated CHSH game that satisfies Condition~\ref{cond:self-test_all_N}.

\begin{lemma}\label{lem:comp_top_sg}
    The closure $\Ess$ of the semigroup $S$ is a compact topological semigroup.
\end{lemma}

\begin{proof}
	The elements of $\Ess$ are linear maps on a finite-dimensional vector space. 
    Hence their composition is continuous.
	Therefore $\Ess$ is a topological semigroup.
    $\Ess$ is closed by construction.
	The elements of $\Ess\hat\rho$ are all states, as their pairing with $\1^{\otimes 2}$ is $1$.
	Hence $\Ess\hat\rho$ is a closed subset of the state space $\S^{(2)}$, which is a compact set (c.f.\ App.~\ref{app:app_to_formalism}).
    Thus $\Ess\hat\rho$ is a compact set.
	Finally, $\Ess$ is simply the image of $\Ess\hat\rho$ under the right action of $\hat\rho^{-1}$, which implies $\Ess$ is compact.
\end{proof}

To proceed, we need the following result which holds for compact topological semigroups (in general) and hence in particular for $\Ess$.
The result is a direct corollary of Theorem~16.3 and Lemma~16.1 of \cite{eisner2015operator}.

\begin{theorem}\label{thm:ellis}
	The semigroup $\Ess$ contains a minimal right-ideal which, in turn, contains at least one idempotent $P$.
	Moreover, $\Gee_P:=P \Ess P$ is a compact topological group, with neutral element $P$.
\end{theorem}

For the remainder of the discussion, we fix one such idempotent $P$.
The final result will not depend on which one has been chosen.

\begin{lemma}\label{lem:compact_grp_has_pos_char}
    The character $\chi_P$ afforded by $\Gee_P$ is non-negative. 
\end{lemma}

\begin{proof}
    For every outcome $\vec k$, $\tfrac1{p_{\vec k}}R_{\vec k} \hat\rho$ is an element of the bipartite state space $\S^{(2)}$.
    Every $g \in \Gee_P$ arises as the limit of operations of the form $\tfrac1{p_{\vec k}}R_{\vec k}$. 
    Thus, since $\S^{(2)}$ is closed, it follows that $g\hat\rho$ is also an element of $\S^{(2)}$.
    Therefore, fixing any $\phi_k \in \M$ we have
	\begin{align*}
        \tr\big(g \hat\rho\, \hat\phi_k \big) \geq 0.
	\end{align*}
    Using the fact that $P$ is the identity of $\Gee_P$ we get 
    \begin{align*}
        \tr\big(g \hat\rho \, \hat\phi_k \big)
		&
		=  \tr\big((P g P) \, R_k \big)
		= \tr\big(g \, (P R_k P) \big).
    \end{align*}
    But $P R_k P$ is of the form $p_k\, h$ for some $h \in \Gee_P$.
    Thus, for any $g'\in \Gee$, choosing $g:=g' h^{-1}$ gives
    \begin{align*}
        \chi_P(g') 
        = \chi_P(g h) 
        = p_k \tr(g P R_k P) \geq 0.
    \end{align*}
\end{proof}

Now we turn our attention back to the correction map $\vartheta : S \to D_4$.
The map may be extended to $\Ess$ by continuity.

Indeed,
if a pair $\mathcal{R}, \mathcal{R}' \in S$ are sufficiently close together,
then the respective corrections $\vartheta, \vartheta'$ must be the same.
This holds because there is a unique correction that attains the CHSH value ($4a$),
while any other correction will cause the CHSH value to deviate from its maximum by at least $2a>0$.
But the CHSH value is a continuous function on bipartite states and hence a continuous function of $\mathcal{R}$.

\emph{Remark:} It follows that $\Ess$ decomposes into a set of disconnected components, on which the correction map is constant.    

Restricting this map to the compact group $\Gee_P \subseteq \Ess$, we obtain a group homomorphism from $\Gee_P$ into $D_4$.
We will use the same letter, $\vartheta$, for the map $\Gee_P\to D_4$.

Since the image of $\vartheta$ is $H \subseteq D_4$, we have
\begin{align}\label{eqn:correction_grp_def}
    \Gee_P / \ker\vartheta \cong H.
\end{align}
Let $\mu$ be the Haar measure on $\ker \vartheta$.
Then 
\begin{align}\label{eqn:haar}
     \Pi_\varphi := \int_{g \in \ker \vartheta}d\mu(g)\,g
\end{align}
is a projection onto the trivial representation of the normal subgroup $\ker\vartheta \trianglelefteq \Gee_P$. 
Let $g \in \Gee_P$ and define the map
\begin{align*}
    \varphi : H \to \L(\img \Pi_\varphi),\ \vartheta(g) \mapsto \Pi_\varphi g \Pi_\varphi.
\end{align*}

\begin{theorem}\label{thm:the_trace_positive_rep}
	The map $\varphi$ is a representation of the group $H$ with the following properties:
    \begin{enumerate}
				\item 
					Its character $\chi_\varphi$ is non-negative.
        \item 
					The trivial representation has multiplicity 1 in the irrep decomposition of $\varphi$.
    \end{enumerate}
\end{theorem}

Both properties relate directly to the physics of the situation.
Regarding the first point, it turns out that the values of the character can be expressed as a pairing between states and effects, and is thus non-negative.
The proof of the second fact relates the rank of the projection onto the trivial representation to the rank of 
the marginal state shared by Alice and Charlie.
The marginal state is of rank one because it factorizes by construction.

\begin{proof}
	The map $\varphi$ is the restriction of $\Gee_P$ to the subspace that is invariant under $\ker \vartheta \trianglelefteq \Gee_P$.
	Arguing as in the proof of Lemma~\ref{lem:compact_grp_has_pos_char}, it holds that $g\rho\in \S^{(2)}$ for every $g$ that appears in the integral in Eq.~(\ref{eqn:haar}).
	Because the Haar measure is normalized and $\S^{(2)}$ a convex set, 
	$\Pi_\varphi \hat\rho$ is a state.
	In particular, this means that $\img(\Pi_\varphi)$ is non-trivial.
	Thus it is a representation of the quotient group $H$.

    For (1) we first choose a transversal of $\Gee_P/\ker \vartheta$ in order to compute the character.
    Indeed by Lem.~\ref{lem:coarse-graining}, every element of $\M$ correspond to a unique correction in $H$, and thus the set
    \begin{align*}
        \big\{ g_k := P\tfrac1{p_k}R_kP \ | \ k \in [n] \big\}
    \end{align*}
    is in one-one correspondence with the elements of $H$, and hence is a transversal of $\Gee_P/\ker\vartheta$.
    Now, for any $k \in [n]$ we can compute
    \begin{align*}
        \chi_\varphi(k) 
        = \tr(\Pi_\varphi g_k \Pi_\varphi) 
		&
        = \tr(g_k \Pi_\varphi)
        \\
		&
        = \int_{l \in \ker\vartheta}d\mu(l)\chi_P(g_k l) \geq 0,
    \end{align*}
    where we have used that the character $\chi_P$ is non-negative (Lem.~\ref{lem:compact_grp_has_pos_char}).

    We prove (2) by bounding the multiplicity of the trivial representation above and below by one.
    The lower bound always holds for a non-negative character of a finite group.
    This is because the character inner product of any non-negative character $\chi$ with the trivial character $\chi_1$ yields 
    \begin{align*}
        [\chi,\chi_1] 
        = \tfrac1{|G|} \sum_{g \in G} \chi_1(g)^*\chi(g)
        = \tfrac1{|G|} \sum_{g \in G} \chi(g) > 0,
    \end{align*}
    since $\chi(\Id) > 0$.

    For the upper bound, we show that (i) The trivial representation is a subspace of the $+1$ eigenspace of the map
    \begin{align*}
        \sum_{k \in [n]} p_k \Pi_\varphi g_k \Pi_\varphi,
    \end{align*}
    and that (ii) The rank of this map is upper bounded by 1.
    Indeed, by definition, for any $v$ in the trivial representation, we have $\forall k \in [n], \Pi_\varphi g_k \Pi_\varphi \, v = v$, and thus 
    \begin{align*}
        \bigg(\sum_{k \in [n]} p_k \Pi_\varphi g_k \Pi_\varphi \bigg) \, v 
        = \sum_{k \in [n]} p_k \, v 
        = v,
    \end{align*}
    hence (i) follows.
    For (ii), using $\Pi_\varphi P = P \Pi_\varphi = \Pi_\varphi$, which is a consequence of the definition of $\Pi_\varphi$ (Eqn.~(\ref{eqn:haar})), and the fact that $\forall g \in \Gee_P, PgP = g$ and $P^2 = P$, 
    we expand the map: 
    \begin{align*}
        \sum_{k \in [n]} p_k \Pi_\varphi g_k \Pi_\varphi
        &
        = \sum_{k \in [n]} p_k \Pi_\varphi P\tfrac1{p_k}R_kP \Pi_\varphi
        \\
        &
        = \sum_{k \in [n]} \Pi_\varphi R_k \Pi_\varphi
        \\
        &
        = \Pi_\varphi \bigg(\sum_{k \in [n]}  R_k \bigg) \Pi_\varphi
        \\
        &
        = \Pi_\varphi \bigg(\sum_{k \in [n]} \hat\rho\hat\phi_k \bigg) \Pi_\varphi
        \\
        &
        = \Pi_\varphi \hat\rho \bigg(\sum_{k \in [n]} \hat\phi_k \bigg) \Pi_\varphi
        \\
        &
        = \Pi_\varphi \big(\rho(\1, \cdot) \1 \big) \Pi_\varphi.
    \end{align*}
    The map
    \begin{align*}
        \rho(\1, \cdot) \1 : V_C^* \to V_C^*,\  \sigma \mapsto \sigma(\1) \rho(\1, \cdot)
    \end{align*}
    is manifestly of rank one, and hence $\Pi_\varphi \big(\rho(\1, \cdot) \1 \big) \Pi_\varphi$ is of rank at most one.
    Thus (ii) and (2) follow.
\end{proof}

There are two dual ways to describe entanglement swapping experiments:
Either 
as teleporting $C$-type states,
or 
as teleporting $A$-type effects
(c.f.\ Fig.~\ref{fig:right-to-left}).
The map $R_k$ used so far formalizes the former point of view.
For the latter, define
\begin{align*}
	L_k := \hat\phi_k \hat\rho : V_A \to V_A.
\end{align*}
Extending this definition to $L_{\vec k}$ (analogously to $R_{\vec k}$), we find that $L_{\vec k}$ and $R_{\vec k}$ are conjugate to each other with respect to $\hat\rho$:
\begin{align*}
    L_{\vec k} = \hat\rho^{-1} R_{\vec k} \hat\rho.
\end{align*}
This allows us to replicate the results for subsystems of $C$-type, on subsystems of $A$-type as follows:
Conjugating $\Ess$ by $\hat\rho$ gives us Lem.~\ref{lem:comp_top_sg} for the semigroup generated by elements of the form $\tfrac1{p_k}L_k$.
Next, since $P$ is obtained as a convergent sequence of elements of the form $\tfrac1{p_{\vec k}} R_{\vec k}$, the map $\hat\rho^{-1}P\hat\rho$ is a convergent sequence of elements of the form $\tfrac1{p_{\vec k}} L_{\vec k}$.
Thus we also get Thm.~\ref{thm:ellis} on subsystems of $A$-type, where the compact group is $\hat\rho^{-1}\Gee_P\hat\rho$.
The rest of the results -- Lem.~\ref{lem:compact_grp_has_pos_char} and Thm.~\ref{thm:the_trace_positive_rep} -- follow by invariance under conjugation.
Thus we get a representation of $H$ on the subsystems of $A$-type equivalent to the representation $\varphi$ as follows: For $h \in H$
\begin{align*}
    \pi(h) := \hat\rho^{-1} \varphi(h) \hat\rho.
\end{align*}
The corresponding representation space is supported on the image of the projection $\Pi_\pi := \hat\rho^{-1} \Pi_\varphi \hat\rho$.

\begin{figure}[h]
    \centering
    \includegraphics[trim={0 3cm 0 3cm}, clip, width=\linewidth, keepaspectratio]{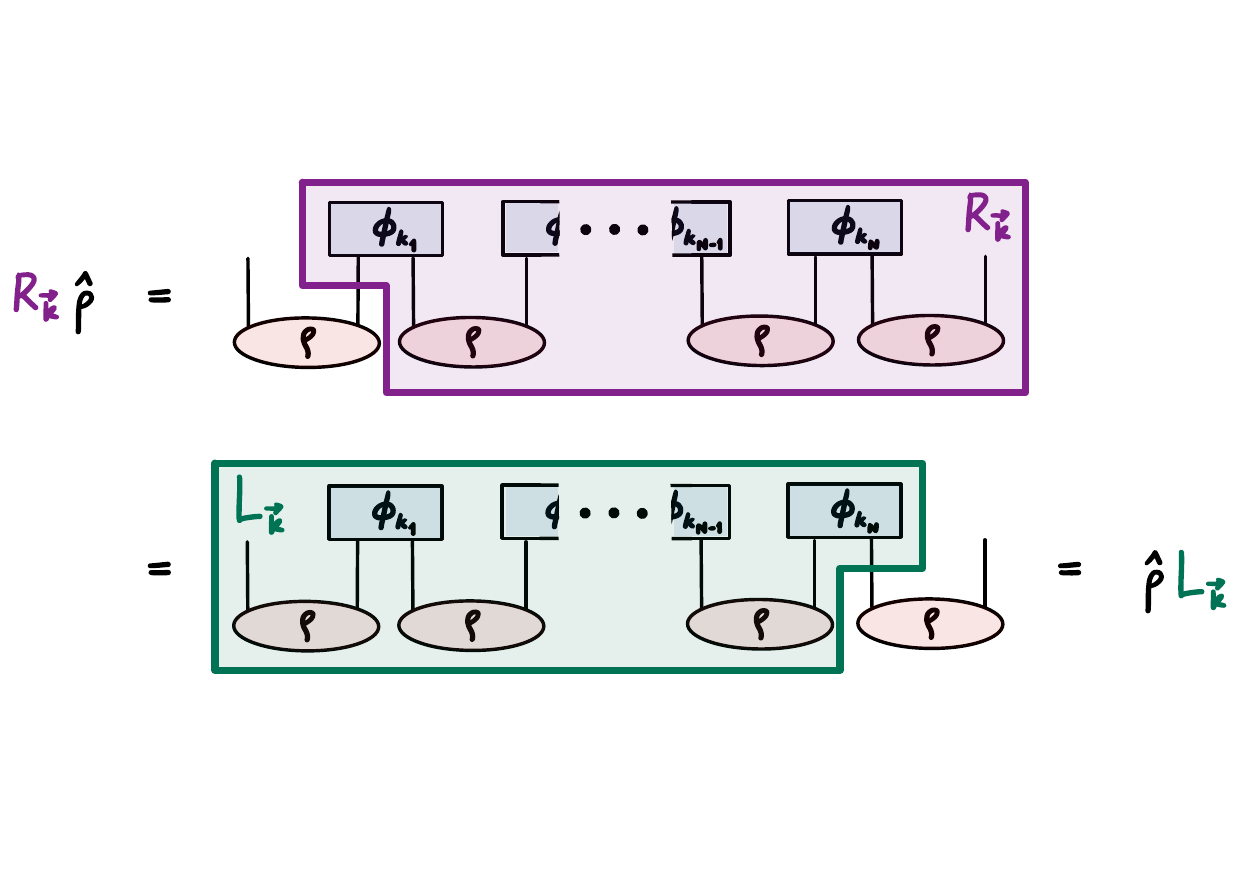}
    \caption{The two ways to read the result of entanglement swapping.
    One can either read it as first applying the map $\hat\rho : V_A \to V_C^*$ and then the map $R_{\vec k} : V_C^* \to V_C^*$ or first the map $L_{\vec k} : V_A \to V_A$ and then the map $\hat\rho : V_A \to V_C^*$.}
    \label{fig:right-to-left}
\end{figure}

Now that we have a representation space both on $A$-type and $C$-type subsystems, we can construct the effective GPT of the iterated CHSH game.


\subsection{The effective Iterated CHSH GPT}\label{subsec:eff_I_CHSH_GPT}

Mimicking the approach of Sec.~\ref{sec:GPT_self-testing}, 
here we define the \emph{effective Iterated CHSH GPT}
as the GPT generated by the states and effects that are required to realize the iterated CHSH game.
A representation-theoretic analysis will show that the effective GPTs that arise this way fall into one of seven inequivalent families.

\begin{definition}[The effective I-CHSH GPT]\label{def:eff_I-CHSH_GPT}
    Let $\rho, \phi_k, e_i, f_j$ be an instance of the iterated CHSH game satisfying Condition~\ref{cond:self-test_all_N}.
    Let $H$ be the associated correction group (Eqn.~(\ref{eqn:correction_grp_def})).
    Let $\Omega_A := \{ e_i, \neg e_i\}_{i = 0,1}$, $\Omega_C := \{ f_j, \neg f_j\}_{j = 0,1}$.
    Let $\tilde\rho$ be the bilinear form associated with the map $\Pi_\varphi \hat\rho \Pi_\pi$.
    Define the sets 
    \begin{itemize}
        \item $\tilde \Omega_A := \Pi_\pi \Omega_A,\ \tilde \Omega_C := \Pi_\phi^t \Omega_C$,
        \item $\Psi := \big\{ \varphi(h)\Pi_\varphi \hat\rho \Pi_\pi \big\}_{h \in H}$,
        \item $\Phi := \big\{ \Pi_\pi \hat\phi_k \Pi_\varphi \big\}_{k \in [n]}$.
    \end{itemize}
	Then the \emph{effective I-CHSH GPT} is the GPT closure (in the sense of Sec.~\ref{subsubsec:generating_a_GPT}) 
    of the input: $P = \cone\big(\Phi \cup \tilde\Omega_C \otimes \tilde\Omega_A \big)$, $D = \cone\big( \Psi \cup \tilde\rho(\cdot, \tilde\Omega_C) \otimes \tilde\rho(\tilde\Omega_A, \cdot) \big)$.
\end{definition}

\emph{Remarks:}
(1) In the above definition, $D \subseteq \img(\Pi_\pi)^* \otimes \img(\Pi_\varphi)$ and $P \subseteq \img(\Pi_\varphi)^* \otimes \img(\Pi_\pi)$. 
(2) effective I-CHSH GPT will turn out to be \emph{not} locally tomographic. 
We discuss this further in Sec.~\ref{subsec:no_all_pancakes}.


\subsection{The seven families}\label{subsec:5+2_families}

The representation $\varphi$ of the correction group $H$,
and hence the effective I-CHSH GPT,
are subject to the following constraints:
\begin{enumerate}
    \item $\varphi$ has dimension at least three. 
    This holds because we need local dimension at least 3 to violate CHSH.
	\item The character of the representation $\chi_\varphi$ has to be non-negative (Thm.~\ref{thm:the_trace_positive_rep}).
	\item The trivial character $\chi_1$ occurs exactly once in the decomposition of $\chi_\varphi$ (Thm.~\ref{thm:the_trace_positive_rep}).
\end{enumerate}

The correction group $H$ is one of the following groups: The trivial group, $\ZZ_2$ (the cyclic group of order two), $\ZZ_4$ (the cyclic group of order four), $K_4$ (the Klein four-group), and $D_4$ (the dihedral group of order eight).
Of these, we can rule out the trivial group and $\ZZ_2$ for the following reasons:
\begin{itemize}
    \item The only irreducible character afforded by any representation of the trivial group is the trivial character $\chi_1$.
    Due to constraint~(3), $\chi_\varphi = \chi_1$ is the only option. 
    This representation is one-dimensional and hence ruled out by constraint~(1).

    \item The irreducible characters of $\ZZ_2$ are the trivial character $\chi_1$ and $\chi_2 = (1,\ -1)$.
    Thus the only possibilities given constraints~(2)~and~(3) are $\chi_\varphi = \chi_1$ and $\chi_\varphi = \chi_1 + \chi_2$.
    Both are ruled out by constraint~(1).
\end{itemize}

The remaining groups -- $\ZZ_4$, $K_4$ and $D_4$ -- all contribute solutions.
The character tables for these groups are recorded in App.~\ref{app:D_4}.

Let $\Xi$ represent the character table of any of the three groups.
The above three constraints correspond to the following system of Diophantine inequalities in the rows of $\Xi$:
Let $r = \mathrm{rows}(\Xi)$, and $n_i \in \NN$ for $i = [r]$, then 
\begin{enumerate}
    \item $\sum_{i = 1}^r n_i \, \Xi_{i1} \geq 3$,
    \item $\forall j \in [r],  \sum_{i = 1}^r n_i \Xi_{ij} \geq 0$,
    \item $n_1 = 1$.
\end{enumerate}
In the present case, this system of inequalities was solved using the SageMath computer algebra system \cite{sagemath} (the computer code can be found here \cite{zenodo}).
The result is as follows:

\begin{result}\label{res:classification}
    Let $\{ \chi_{i}^{(H)} \}_i$ (refer App.~\ref{app:D_4}) be the irreducible characters respectively of $H = \ZZ_4, K_4$, and $D_4$.
    Then, using the convention 
    \begin{align*}
        \chi_{(i_1)^{j_1} \cdots (i_n)^{j_n}} = j_1\chi_1 + \cdots + j_n\chi_n,
    \end{align*}
    $\chi_{\varphi}$ has to be one of the following seven characters:
    \begin{itemize}
        \item $\dim = 4 : \chi_{1234}^{(\ZZ_4)}, \chi_{1234}^{(K_4)}, \chi_{125}^{(D_4)}, \chi_{135}^{(D_4)}, \chi_{145}^{(D_4)}$, where the first two are the regular character of $\ZZ_4$ and $K_4$ respectively;
        \item $\dim = 6 : \chi_{12345}^{(D_4)}$;
        \item $\dim = 8 : \chi_{12345^2}^{(D_4)}$, the regular character of $D_4$.
    \end{itemize}    
    In $\dim = 4$, the computer algebra system reports an additional solution, namely $\chi_{1234}^{(D_4)}$.
    The representation affording this character is isomorphic to the one affording $\chi_{1234}^{(K_4)}$, and thus is not listed above.
\end{result}

Result~\ref{res:classification} states that any effective I-CHSH GPT necessarily falls into one of seven families.
It is not a priori clear that all seven of these families are inhabited, i.e., it is possible to construct an example of each type.
Resolving this question in the affirmative, we construct explicit GPTs for each family in Appendix~\ref{App:family_members}.


\subsection{Quantum family members}\label{subsec:Q_family}

In this section we collect some positive and negative results regarding which of the families in Result~\ref{res:classification} can be realized within quantum theory.

In quantum mechanics, (i) one can realize an iterated CHSH game that maintains QM's CHSH value of $2\sqrt2$.
Moreover, (ii) any such realization satisfies Condition~\ref{cond:self-test_all_N}.
Indeed, for (i), consider the instance where Alice measures in Pauli $X$ and $Z$ bases, Charlie measures in bases that is rotated by $\pi/4$ in the Pauli $X$--$Z$-plane, the shared bipartite state is the Bell state
\begin{align*}
    |\Phi^+\rangle = \tfrac1{\sqrt2} (|00\rangle + |11\rangle),
\end{align*}
and the measurement $\M$ is the Bell basis measurement, i.e., the measurement in the basis $\{ \sigma_k \otimes \Id |\Phi^+\rangle \}_{k = 0, \dots, 3}$, with $\sigma_0 = \Id$ and $\sigma_k, \ k = 1, \dots, 3$ being Pauli $X,Y$ and $Z$ respectively.

In this instance, the bipartite state conditioned on outcome $k$ of $\M$ corresponds to the the element of the Bell basis labelled by $k$.
Every element of the Bell basis can be mapped back to $|\Phi^+\rangle$ by Alice, by applying the appropriate Pauli correction.
This correction corresponds to an action of $\sigma_k$ by conjugation on Alice's measurement effects.
Pauli operators acting by conjugation correspond to a $K_4$ group action, which can be realized by relabelling Alice's measurement apparatus.
After the relabelling, the situation is simply that of the usual Bell test.

Alternatively, if we are given a quantum instance of the iterated CHSH game which achieves a CHSH-value of $2\sqrt2$ for every $N \in \NN$ and every outcome $\vec k$, then point (ii) follows quantum self-testing of CHSH-value $2\sqrt2$.
This is because the conditions on the observables $A_i, C_j$ due to self-testing imply the conditions of Lem.~\ref{lem:sufficient_st_conditions} (c.f.\ \cite{vsupic2020self}).

The above instance of the iterated CHSH game is a member of the family $\chi^{(K_4)}_{1234}$.
What is yet unclear is whether we can construct members of other families within QM.
Recall that by Wigner's Theorem \cite{wigner2012group, BargmannWigner}, a symmetry group acting on density operators must result from a projective unitary or anti-unitary representation on Hilbert space.
Anti-unitaries are not completely positive, and therefore not physically implementable.
Hence we can restrict attention to symmetry groups that allow for a projective unitary representation on the underlying Hilbert space.
As a consequence, we have

\begin{lemma}\label{lem:no_Z4}
    The family labelled by $\chi^{(\ZZ_4)}_{1234}$ cannot be realized within QM.
\end{lemma}

\begin{proof}
	Assume for the sake of reaching a contradiction that it is possible to realize members of the $\chi^{(\ZZ_4)}_{1234}$ family within quantum theory.
    Then the corresponding representation of $\ZZ_4$ must come from a projective unitary representation on Hilbert space.
	Every projective representation of a cyclic group is projectively equivalent to a linear representation \cite[Theorem~2.3.1]{Karpilovsky1985} (originally \cite{schur1904}). 
	In particular, the linear representation on Hilbert space is abelian, and thus admits a common eigenbasis $\{ |\psi_i\rangle \}_i$.
	Hence, at the level of density matrices, the projections $|\psi_i\rangle\langle\psi_i|$ are invariant under the $\ZZ_4$ action.
	The basis has at least two elements, because no Bell inequalities can be violated with a one-dimensional local Hilbert space.
    This implies that the trivial representation has multiplicity at least two in the decomposition of the representation on the density matrices, contradicting the fact that the trivial representation has multiplicity 1 (Thm.~\ref{thm:the_trace_positive_rep}~(2)).
\end{proof}

Since the group $K_4$ is also abelian, it's trivial projective class would be subject to the same restrictions as $\ZZ_4$.
However, $K_4$ also admits a non-trivial projective class.
This non-trivial projective class is represented on Hilbert space by the Pauli matrices (c.f.\ \cite[Section~3.6]{Karpilovsky1985}), and accounts for the family $\chi^{(K_4)}_{1234}$.

It is also possible to realize the family $\chi^{(D_4)}_{125}$ QM.
However, since the corresponding representation is four-dimensional and the group $D_4$ has order eight, 
this family cannot be realized as a basis measurement, in contrast to $\chi^{(K_4)}_{1234}$. 
However it is possible to realize it as a POVM.
Let $| b_k \rangle = \sigma_k \otimes \Id |\Phi^+\rangle$ and $| a_k \rangle = S\sigma_k \otimes \Id |\Phi^+\rangle$, where $S$ is the phase operator with $S^2 = \sigma_z = \sigma_3$.
Then the choosing $\M$ to be the following POVM
\begin{align*}
    \M := 
    \big\{ \tfrac12 | b_k \rangle\langle b_k | \big\}_{k = 0, \dots, 3}
    \cup
    \big\{ \tfrac12 | a_k \rangle\langle a_k | \big\}_{k = 0, \dots, 3},
\end{align*} 
along with the Bell state and the Pauli measurements of the standard Bell test (noted in the beginning of this section, Sec.~\ref{subsec:Q_family}), realizes a member of $\chi^{(D_4)}_{125}$.
The question of whether the rest of the families have a quantum realization is left open.


\subsection{GPT pancakes}\label{subsec:no_all_pancakes}

A conceptual consequence of Result~\ref{res:classification} is that iterated entanglement swapping games can give an operational justification for constructing GPTs that violate the \emph{local tomography} postulate.

Local tomography is a commonly made assumption \cite{hardy2001quantumtheoryreasonableaxioms,hardy2013reconstructingquantumtheory,barrett2007information,chiribella2017quantum,Muller2021reconstruction} (also discussed in \cite{barnum2014localtomography,lismer2025experimentaltestprincipletomographic,Galley2018anymodificationof}) that states that in a probabilistic theory, the product effects should span the space of all effects.
In this case, a state can be characterized from the result of product measurements alone, hence the name.

The effective I-CHSH GPT (Def.~\ref{def:eff_I-CHSH_GPT}) and all the examples constructed in App.~\ref{App:family_members} violate this principle.
Their local effect space is three-dimensional (spanned by the observables required to implement CHSH tests),
but the characterization in Result~\ref{res:classification} lists no solution with three local dimensions.
Intuitively, the reason is that the space spanned by the tensor product of the effects needed to specify the CHSH experiment 
is ``too small'' to support the bipartite effects that are used in entanglement swapping.
Thus the violation of local tomography is justified by the operational requirement of being able to perform entanglement swapping.

In the context of qubit quantum mechanics, there is a theorem, colloquially called the ``no-pancake theorem'' \cite{BETHRUSKAI2002159, nopancake2005kohout}, that states that there is no completely positive map whose domain is the Bloch ball and whose image is contained in a plane (i.e., a ``pancake'').
The elements of the teleportation semigroup of a GPT are required to be completely positive on the GPT by consistency.
Thus in analogy to the QM case, under the assumption of local tomography, we can conclude a similar result:

\begin{lemma}[GPT no-pancake theorem]\label{lem:GPT_no_pancake}
    In any locally tomographic GPT, which contains at least one instance of the iterated CHSH game satisfying Condition~\ref{cond:self-test_all_N}, the directional vector space of the image of the local state space under the teleportation semigroup of the GPT must be of dimension at least three.
\end{lemma}

In particular, Lem.~\ref{lem:GPT_no_pancake} tells us that there are no locally tomographic GPTs with local dimension three or less (recall that the local dimension is one more than the dimension of the directional vector space of the state space).
This gives us an alternative proof of the fact that it is not possible to perform entanglement swapping in boxworld theory, or any of the regular polygon GPTs of \cite{janotta2011limits}.


\section{Conclusions and Outlook}

The initial motivation of this work came from Refs.~\cite{weilenmann2020self, weilenmann2020towards}.
Their premise is that ``sustaining a CHSH value after teleportation'' is a rare property for probabilistic theories to have -- potentially so rare that an operational characterization of quantum correlations could be based on them.
In Ref.~\cite{dmello2024entanglement}, we provided initial evidence against this conjecture, by constructing a post-quantum GPT with this property.

In the present work, we show that the requirement that a theory sustains its CHSH value after teleportation does highly constrain its structure.
Maybe surprisingly, this invariance property gives rise to a representation-theoretic condition, all solutions of which can be explicitly enumerated.

In this sense, the property of Refs.~\cite{weilenmann2020self, weilenmann2020towards} \emph{is} rare after all -- just not sufficiently so to single out QM without further assumptions.

Beyond that,
our ``seven-fold way'' classification
provides an operational justification for rejecting local tomography
as an axiom for probabilistic theories:
The space spanned by the product effects needed to realize the CHSH test is insufficient to support the effects required for entanglement swapping.

One of the tools required to obtain our results was a notion of self-testing in the framework of GPTs, which we have also introduced in this work.
GPT self-testing is a generalization of quantum self-testing to the GPT framework.
It involves making uniqueness statements about the theory obtained by discarding all degrees of freedom that are irrelevant to the situation at hand.
It would be interesting to explore this approach further, and identify other scenarios where GPTs can be self-tested.

Finally, we have left open the question exactly which of the seven family members can be realized within QM.


\section{Acknowledgements}

We are grateful to 
Marc-Olivier Renou for motivating this project, 
Johan {\AA}berg for all the great questions and feedback throughout the course of this work, 
Joachim Krug for providing insights on semigroups, 
and all of Elie Wolfe, Rob Spekkens, Lucas Tendrik, Markus Methlinger, Martin Renner, Nicolas Brunner, Pavel Sekatski, Sadra Boreiri and Xiangling Xu for helpful discussion. 

Parts of this work was conducted while one of us (LD) was kindly hosted by the
Perimeter Institute for Theoretical Physics. 
We acknowledge support by Germany’s Excellence Strategy -- Cluster of
Excellence Matter and Light for Quantum Computing (ML4Q), EXC 2004/1
(390534769).


\bibliography{bibliography}


\onecolumngrid

\section{Appendix}

\subsection{The dihedral group of four elements}\label{app:D_4}

The content of this appendix is based on the references \cite{Isaacs1976}, \cite{Simon1996}, and \cite{Robinson1996}.
The details were computed by hand.

The dihedral group of four elements $D_4$, is the symmetry group of the square.
It is realized as the semi-direct product of two cyclic groups, namely
\begin{align*}
    D_4 = \langle \xi \rangle \rtimes \langle \eta \rangle \cong \ZZ_4 \rtimes \ZZ_2.
\end{align*}
Where $\xi^4 = \Id$ and $\eta^2 = \Id$, and $\eta \xi^k \eta = \xi^{-k}$.
What will be an important feature in our discussion is that $D_4$ is isomorphic to the wreath product of $\ZZ_2$ with itself, i.e.,
\begin{align*}
    D_4 \cong \ZZ_2 \wr \ZZ_2.
\end{align*}
The conjugacy classes($\Kay$) and corresponding centralizers($\Cee$) of $D_4$ are as follows:
\begin{gather*}
    \Kay_{\Id} = \{\Id\}, \quad
    \Kay_{\xi}  = \{\xi, \xi^3\}, \quad
    \Kay_{\xi^2}  = \{\xi^2\}, \quad
    \Kay_{\eta}  = \{\eta, \xi^2\eta\}, \quad
    \Kay_{\xi \eta}  = \{\xi\eta, \xi^3\eta\}.
    \\[1em]
    \Cee_{\Id} = D_4, \quad 
    \Cee_{\xi} = \langle \xi \rangle, \quad 
    \Cee_{\xi^2} = D_4, \quad
    \Cee_{\eta} = \langle \{ \xi^2, \eta \} \rangle, \quad 
    \Cee_{\xi\eta} =  \langle \{\xi^2, \xi\eta\} \rangle. 
\end{gather*}
There are $5$ conjugacy classes.
This means the number of representative irreps of the group algebra $\CC[D_4]$ are $|\Em(\CC[D_4])| = 5$.
Also the dimension of the group algebra $\dim \CC[D_4] = 8$.
Therefore we have $\Em(\CC[D_4]) = \{M_i\}_{i = 1, \cdots, 5.}$, and
\begin{align*}
    \sum_{i = 1}^{5} (\dim M_i)^2 = 8.
\end{align*}
The only solution is $1^2 + 1^2 + 1^2 + 1^2 + 2^2$, i.e., there is one $2D$ irrep and four linear irreps.
This information is sufficient to deduce the \emph{character table} of $D_4$, which is as follows:
\begin{table}[H]
    \centering
    \renewcommand{\arraystretch}{1.5}
    \begin{tabular}{|c|ccccc|c|}
        \hline
        $\Kay_{D_4}$ & $\Kay_{\Id}$ & $\Kay_{\xi}$ & $\Kay_{\xi^2}$ & $\Kay_{\eta}$ & $\Kay_{\xi\eta}$ & \multirow{3}{*}{irrep} \\
        \cline{1-6}
        $|\Kay_{D_4}|$ & $1$ & $2$ & $1$ & $2$ & $2$ & \\
        \cline{1-6}
        $|\Cee_{D_4}|$ & $8$ & $4$ & $8$ & $4$ & $4$ & \\
        \hline 
        $\chi_1$ & $1$ & $1$ & $1$ & $1$ & $1$ & trivial \\
        $\chi_2$ & $1$ & $1$ & $1$ & $-1$ & $-1$ & $D_4/\ZZ_4$ \\
        $\chi_3$ & $1$ & $-1$ & $1$ & $1$ & $-1$ & $(-1)^{\pi(g)}$ \\
        $\chi_4$ & $1$ & $-1$ & $1$ & $-1$ & $1$ & $D_4/\ZZ_4 \otimes (-1)^{\pi(g)}$ \\
        $\chi_5$ & $2$ & $0$ & $-2$ & $0$ & $0$ & planar roto-reflections \\[1ex]
        \hline
    \end{tabular}
    \caption{The character table of $D_4$}
    \label{tab:D_4_ch_tab}
\end{table}
The subgroups of $D_4$ are
\begin{align*}
    \O(1)
    &
    :\Id;
    \\
    \O(2)
    & 
    :\langle \xi^2 \rangle, 
    \langle \eta \rangle,
    \langle \xi\eta \rangle,
    \langle \xi^2\eta \rangle,
    \langle \xi^3\eta \rangle;
    \\
    \O(4)
    &
    :\langle \xi \rangle,
    \langle \{ \xi^2, \eta\} \rangle,
    \langle \{\xi^2, \xi\eta \} \rangle;
    \\
    \O(8)
    &
    :D_4.
\end{align*}

Out of these the order $4$ subgroups, namely $K_4$ and $\ZZ_4$, are of interest.
Their character tables are as follows:
\begin{table}[H]
    \centering
    \begin{minipage}{0.45\textwidth}
        \centering
        \renewcommand{\arraystretch}{1.5}
        \begin{tabular}{|c|cccc|}
            \hline
            $\Kay_{K_4}$ & $\Kay_{\Id}$ & $\Kay_{\xi^2}$ & $\Kay_{\eta}$ & $\Kay_{\xi^2\eta}$ \\
            \hline 
            $\chi_1$ & $1$ & $1$ & $1$ & $1$ \\
            $\chi_2$ & $1$ & $1$ & $-1$ & $-1$ \\
            $\chi_3$ & $1$ & $-1$ & $1$ & $-1$ \\
            $\chi_4$ & $1$ & $-1$ & $-1$ & $1$ \\[1ex]
            \hline
        \end{tabular}
        \caption{The character table of $K_4$}
        \label{tab:K_4_ch_tab}
    \end{minipage}
    \hfill
    \begin{minipage}{0.45\textwidth}
        \centering
        \renewcommand{\arraystretch}{1.5}
        \begin{tabular}{|c|cccc|}
            \hline
            $\Kay_{\ZZ_4}$ & $\Kay_{\Id}$ & $\Kay_{\xi}$ & $\Kay_{\xi^2}$ & $\Kay_{\xi^3}$ \\
            \hline 
            $\chi_1$ & $1$ & $1$ & $1$ & $1$ \\
            $\chi_2$ & $1$ & $i$ & $-1$ & $-i$ \\
            $\chi_3$ & $1$ & $-i$ & $-1$ & $i$ \\
            $\chi_4$ & $1$ & $-1$ & $1$ & $-1$ \\[1ex]
            \hline
        \end{tabular}
        \caption{The character table of $\ZZ_4$}
        \label{tab:Z_4_ch_tab}
    \end{minipage}
\end{table}


\subsection{Appendix to Sec.~\ref{subsec:GPTs}}\label{app:app_to_formalism}

\begin{lemma*}
    The state space $\S^{(n)}$ is a compact set.
\end{lemma*}

\begin{proof}
    The definition of $\D^{(n)}$ is such that the cone lies entirely in $\Span(\P^{(n)})^*$, where it is a closed, pointed, generating, convex cone (see Sec.~\ref{subsec:GPTs}).
    Thus it is sufficient for the sake of this proof to restrict the ambient spaces of $\D^{(n)}$ and $\P^{(n)}$ to be $\Span(\P^{(n)})^*$ and $\Span(\P^{(n)})$ respectively, and thus assume $\D^{(n)}$ and $\P^{(n)}$ to be closed, pointed, generating, convex cones (see footnotes in Sec.~\ref{subsec:GPTs} for definitions).

    Since we are in finite dimensions it is sufficient to show that $\S^{(n)}$ is closed and bounded.
    The cone $\D^{(n)}$ is closed by definition (Sec.~\ref{subsec:GPTs}) and $\1^{\otimes n}$ is a continuous linear functional.
    $\S^{(n)}$ arises as the intersection of $\D^{(n)}$ and the pre-image of $1$ under $\1^{\otimes n}$ (c.f.\ definition in Sec.~\ref{subsec:GPTs}) and thus is closed.

    By definition $\1^{\otimes n}$ is in the relative interior of $\P^{(n)} \subset (\D^{(n)})'$ (the polar dual of $\D^{(n)}$, see footnote in Sec.~\ref{subsec:GPTs}).
    Therefore, taking the ambient space to be $\Span(\P^{(n)})$ implies that $\1^{\otimes n}$ is in fact in the interior of $\P^{(n)}$.
    As a consequence, it is a strictly positive functional on $\D^{(n)}$, i.e., it is positive on $\D^{(n)} \setminus \{0\}$.
    Assume for the sake of reaching a contradiction that $\1^{\otimes n}$ is positive but not strictly positive on $\D^{(n)}$.
    Then, there exists some $\sigma \in \D^{(n)}$ such that $\sigma(\1^{\otimes n}) = 0$.
    Since $\1^{\otimes n}$ is a continuous functional in the interior of $(\D^{(n)})'$, there exists some perturbation of $p$ of $\1^{\otimes n}$ in the interior of $(\D^{(n)})'$ such that $\sigma(p) < 0$.
    This contradicts the fact that $p \in (\D^{(n)})'$.
    
    Now, let $S_1$ be the unit sphere with respect to the Euclidean norm.
    Because $\D^{(n)}$ is closed, the intersection $\D^{(n)} \cap S_1$ is a compact set.
    Let $\lambda_{\mathrm{min}} > 0$ be the minimum value attained by $\1^{\otimes n}$ on $\D^{(n)} \cap S_1$.
    Then the norm ball $B_{1/\lambda_{\mathrm{min}}}$ of radius $1/\lambda_{\mathrm{min}}$ contains $\S^{(n)}$.
    Thus $\S^{(n)}$ is compact.
\end{proof}


\subsection{Action of the relabelling group on the CHSH observables}\label{app:action_on_chsh_obs}

Denote the standard CHSH observable by 
\begin{align*}
    A_0B_0 + A_0B_1 + A_1B_0 - A_1B_1 \equiv (+ \ + \ + \ -).
\end{align*}
Using this convention we can denote all CHSH observables by
\begin{align*}
    \{ \pm \pi (+ \ + \ + \ -) \ | \ \pi \in S_4 \}.
\end{align*}
For example, exchanging Alice's settings yields the following
\begin{align*}
    &
    A_1B_0 + A_1B_1 + A_0B_0 - A_0B_1
    \equiv A_0B_0 - A_0B_1 + A_1B_0 + A_1B_1
    \equiv ( + \ - \ + \ + ).
\end{align*}
Now, we can define the action of the relabeling group on the CHSH correlators as follows
\begin{align*}
    \neg^{(0)}A_0 = \neg e_0 - e_0 = -A_0, \qquad (s \mapsto \overline s) B_1 = f_0 - \neg f_0 = B_0.
\end{align*}
Using this we can define a group action on the CHSH observable. 
For example,
\begin{align*}
    \neg^{(0)} \otimes \Id (+ \ + \ + \ -) = (- \ - \ + \ -), \qquad \Id \otimes (s \mapsto \overline s)(+ \ + \ + \ -) = (+ \ + \ - \ +)
\end{align*}
Choosing the following realization of $D_4$
\begin{equation}
    \xi := (s \leftrightarrow \overline s)(o \leftrightarrow -o)^{(0)}
    \quad 
    \text{and}
    \quad
    \eta := (o \leftrightarrow -o)^{(0)} \xi,
    \quad 
    \text{with}
    \quad
    D_4 = \langle \xi \rangle \rtimes \langle \eta \rangle,
\end{equation}
we get
\begin{align*}
    e_0
    \overset{\xi}{\mapsto}
    \neg e_1
    \overset{\xi}{\mapsto}
    \neg e_0 
    \overset{\xi}{\mapsto}
    e_1
    \overset{\xi}{\mapsto}
    e_0;
    \qquad
    e_0
    \overset{\eta}{\mapsto}
    \neg e_1;
    \qquad
    e_1
    \overset{\eta}{\mapsto}
    \neg e_0;
    \qquad
    \1 \overset{D_4}{\mapsto} \1;
    \\
    f_0
    \overset{\xi}{\mapsto}
    \neg f_1
    \overset{\xi}{\mapsto}
    \neg f_0 
    \overset{\xi}{\mapsto}
    f_1
    \overset{\xi}{\mapsto}
    f_0;
    \qquad
    f_0
    \overset{\eta}{\mapsto}
    \neg f_1;
    \qquad
    f_1
    \overset{\eta}{\mapsto}
    \neg f_0;
    \qquad
    \1 \overset{D_4}{\mapsto} \1.
\end{align*}
This corresponds the following action on the correlators
\begin{gather*}
    A_0 
    \overset{\xi_A}{\mapsto} 
    -A_1
    \overset{\xi_A}{\mapsto} 
    -A_0
    \overset{\xi_A}{\mapsto}
    A_1 
    \overset{\xi_A}{\mapsto}
    A_0;
    \qquad
    A_0
    \overset{\eta_A}{\mapsto}
    -A_1;
    \qquad
    A_1
    \overset{\eta_A}{\mapsto}
    -A_0;
    \\
    B_0 
    \overset{\xi_B}{\mapsto} 
    -B_1
    \overset{\xi_B}{\mapsto} 
    -B_0
    \overset{\xi_B}{\mapsto}
    B_1 
    \overset{\xi_B}{\mapsto}
    B_0;
    \qquad
    B_0
    \overset{\eta_B}{\mapsto}
    -B_1;
    \qquad
    B_1
    \overset{\eta_B}{\mapsto}
    -B_0;
\end{gather*}
And finally, the following action on the CHSH observables
\begin{align*}
    (\mathrm{Alice}) :
    &
    (+ \ + \ + \ -) 
    \overset{\xi_A}{\mapsto} 
    (+ \ - \ - \ -) 
    \overset{\xi_A}{\mapsto} 
    (- \ - \ - \ +) 
    \overset{\xi_A}{\mapsto} 
    (- \ + \ + \ +) 
    \overset{\xi_A}{\mapsto} 
    (+ \ + \ + \ -);
    \\
    &
    (+ \ + \ + \ -) 
    \overset{\eta_A}{\mapsto}
    (- \ + \ - \ -)
    \overset{\xi_A}{\mapsto}
    (- \ - \ + \ -)
    \overset{\xi_A}{\mapsto}
    (+ \ - \ + \ +)
    \overset{\xi_A}{\mapsto}
    (+ \ + \ - \ +).
    \\[2ex]
    (\mathrm{Bob}) :
    &
    (+ \ + \ + \ -) 
    \overset{\xi_B}{\mapsto} 
    (+ \ - \ - \ -) 
    \overset{\xi_B}{\mapsto} 
    (- \ - \ - \ +) 
    \overset{\xi_B}{\mapsto} 
    (- \ + \ + \ +) 
    \overset{\xi_B}{\mapsto} 
    (+ \ + \ + \ -);
    \\
    &
    (+ \ + \ + \ -) 
    \overset{\eta_B}{\mapsto}
    (- \ - \ + \ -)
    \overset{\xi_B}{\mapsto}
    (- \ + \ - \ -)
    \overset{\xi_B}{\mapsto}
    (+ \ + \ - \ +)
    \overset{\xi_B}{\mapsto}
    (+ \ - \ + \ +).
\end{align*}


\subsection{Representatives of each of the seven families}\label{App:family_members}

Here we provide an explicit member of each of the seven families of GPTs.

(1) For the $\ZZ_4$ representation we will take the correction group to be generated by 
\begin{align*}
    \xi 
    =
    \begin{pmatrix}
        0 & 1 & & \\
        -1 & 0 & & \\
         & & -1 & \\
         & & & 1 \\
    \end{pmatrix}.
\end{align*}
For the local effects, set $r = \sqrt{a \sqrt 2}$, where $a \in (\tfrac12, 1]$ and $4a$ is the CHSH value. 
Then for Alice choose
\begin{gather*}
    \1 = \big( 0 \ 0 \ 0 \ 1 \big)^T, \quad
    e_0 = \tfrac12\big( r \ 0 \ 0 \ 1 \big)^T, \quad
    e_1 = \tfrac12\big( 0 \ r \ 0 \ 1 \big)^T.
\end{gather*}
The effects of Charlie are defined in terms of a $\pi/4$ rotation, i.e., 
\begin{align*}
    f_0 = R e_1, \quad 
    f_1 = R e_0, \quad 
    \text{with} \quad 
    R 
    =
    \begin{pmatrix}
        \tfrac1{\sqrt2} & \tfrac1{\sqrt2} & & \\
        -\tfrac1{\sqrt2} & \tfrac1{\sqrt2} & & \\
         & & 1 & \\
         & & & 1 \\
    \end{pmatrix}.
\end{align*}

Let $\gamma : v \mapsto \langle v, \cdot \rangle$ be the isomorphism induced by the Euclidean inner product.
Define $\rho$ by 
\begin{align}
    \rho(e \otimes f) := \gamma(e)(f),
\end{align}
that is, $\hat\rho = \gamma$.
This gives us 
\begin{gather*}
    \rho(\1 \otimes \1) = 1, \quad
    \rho(e_i \otimes \1) = \rho(\1 \otimes f_j) = \tfrac12, \quad 
    \text{and} \quad
    \rho(e_i \otimes f_j) = \tfrac14\big(1 + (-1)^{ij} a \big).
\end{gather*}
Finally, for the measurement of the Bobs choose 
\begin{align*}
    \M = \{ \hat\phi_k := \tfrac14 \xi^k \gamma^{-1} \ | \ k = 0,1,2,3 \}.
\end{align*}

The construction for the rest of the solution is almost identical.
In fact or all other other $\dim = 4$ representations, the only difference is the matrices that generate the group. 
The rest -- choice of $\rho, e_i, f_j$ -- are identical. 

(2) For $K_4$ the group is generated by the matrices
\begin{align*}
    \xi^2
    =
    \begin{pmatrix}
        -1 &  & & \\
         & -1 & & \\
         & & 1 & \\
         & & & 1 \\
    \end{pmatrix}
    \quad \text{and} \quad
    \eta 
    =
    \begin{pmatrix}
        -1 &  & & \\
         & 1 & & \\
         & & -1 & \\
         & & & 1 \\
    \end{pmatrix}.
\end{align*}
And
\begin{align*}
    \M = \big\{ \hat\phi_g := \tfrac14 g \gamma^{-1} \ | \ g \in \langle \{ \xi^2, \eta \}\rangle \big\}.
\end{align*}

(3) For $\chi_{125}^{(D_4)}$, $\chi_{135}^{(D_4)}$, and $\chi_{145}^{(D_4)}$, the groups are generated respectively by 
\begin{align*}
    \xi_{125}
    =
    \begin{pmatrix}
         0 & 1 & & \\
         -1 & 0 & & \\
         & & 1 & \\
         & & & 1 \\
    \end{pmatrix},
    \quad &\text{and} \quad
    \eta_{125} 
    =
    \begin{pmatrix}
        -1 &  & & \\
         & 1 & & \\
         & & -1 & \\
         & & & 1 \\
    \end{pmatrix};
    \\
    \xi_{135}
    =
    \begin{pmatrix}
         0 & 1 & & \\
         -1 & 0 & & \\
         & & -1 & \\
         & & & 1 \\
    \end{pmatrix},
    \quad &\text{and} \quad
    \eta_{135} 
    =
    \begin{pmatrix}
        -1 &  & & \\
         & 1 & & \\
         & & 1 & \\
         & & & 1 \\
    \end{pmatrix};
    \\
    \xi_{145}
    =
    \begin{pmatrix}
         0 & 1 & & \\
         -1 & 0 & & \\
         & & -1 & \\
         & & & 1 \\
    \end{pmatrix},
    \quad &\text{and} \quad
    \eta_{145} 
    =
    \begin{pmatrix}
        -1 &  & & \\
         & 1 & & \\
         & & -1 & \\
         & & & 1 \\
    \end{pmatrix}.
\end{align*}
Again with 
\begin{align*}
    \M = \big\{ \hat\phi_g := \tfrac18 g \gamma^{-1} \ | \ g \in D_4 \big\}.
\end{align*}

(4) For $\chi_{12345}^{(D_4)}$ and $\chi_{12345^2}^{(D_4)}$ we have to slightly modify the local effects, namely we pad them with enough zeros to make elements of $\RR^6$ and $\RR^8$ respectively.
For example, for $\chi_{12345}^{(D_4)}$ we have:
\begin{gather*}
    \1 = \big( 0 \ 0 \ 0 \ 0 \ 0 \ 1 \big)^T, \quad
    e_0 = \tfrac12\big( r \ 0 \ 0 \ 0 \ 0 \ 1 \big)^T, \quad
    e_1 = \tfrac12\big( 0 \ r \ 0 \ 0 \ 0 \ 1 \big)^T.
\end{gather*}
The matrix $R$ is padded with $1$ on the diagonal to make the dimensions match in order to obtain the $f_j$.
The matrices generating the groups follow the same block diagonal pattern (which can be read off of the character table).
The bipartite state $\hat\rho$ is once again taken to be $\gamma$.
The $\hat\phi_g$ are also chosen the same way.

In all the above examples, we have not specified any effects outside of those required for the CHSH test.
And indeed one can deduce that all these GPTs are \emph{not} locally tomographic.
Of course, they can all be completed to a locally tomographic GPT in a simple way -- pick any ONB $\{w_i\}_i$ of the orthocomplement of the space spanned by the specified local effects, and add the effects $\{ \tfrac12 (w_i \oplus \1) \}_i$ as well as their negations.


\end{document}